\documentclass[a4paper,USenglish,cleveref, autoref, thm-restate]{lipics-v2021}

\pdfoutput=1 
\hideLIPIcs  

\title{Optimal majority rules and quantitative Condorcet properties of setwise  Kemeny voting schemes}  

\titlerunning{The setwise Kemeny problem} 

\author{Xuan Kien Phung}{Département d'informatique et de recherche opérationnelle, Université de Montréal, Montréal, Québec, H3T 1J4, Canada.}{phungxuankien1@gmail.com}{https://orcid.org/0000-0002-4347-8931}{}

\author{Sylvie Hamel}{Département d'informatique et de recherche opérationnelle, Université de Montréal, Montréal, Québec, H3T 1J4, Canada.}{hamelsyl@iro.umontreal.ca}{
https://orcid.org/
0000-0002-8941-2284}{Supported by NSERC through an Individual Discovery Grant RGPIN-2016-04576}

\authorrunning{X.\,K. Phung and S. Hamel} 

\Copyright{X. K. Phung and S. Hamel} 



\category{} 

\relatedversion{} 

\nolinenumbers 

\EventEditors{John Q. Open and Joan R. Access}
\EventNoEds{2}
\EventLongTitle{42nd Conference on Very Important Topics (CVIT 2016)}
\EventShortTitle{CVIT 2016}
\EventAcronym{CVIT}
\EventYear{2016}
\EventDate{December 24--27, 2016}
\EventLocation{Little Whinging, United Kingdom}
\EventLogo{}
\SeriesVolume{42}
\ArticleNo{23}

\usepackage{amsmath, amssymb,  mathabx, amsfonts,enumerate}
\usepackage[all]{xy}
\usepackage{amscd,stmaryrd}
\usepackage{comment}
\usepackage{mathtools}
\usepackage{tikz} 
\usepackage{tikz-cd} 
\tikzcdset{diagrams={nodes={inner sep=7.5pt}}}

\usepackage{url}
\usepackage{hyperref}

\usepackage{array}
\newcolumntype{M}[1]{>{\centering\arraybackslash}m{#1}}

  \theoremstyle{definition}
\newtheorem{question}{Question}

\numberwithin{equation}{section}
 
\newcommand {\N}{\mathbb{N}}

\newcommand\numberthis{\addtocounter{equation}{1}\tag{\theequation}}
\begin{document}

\maketitle

\begin{abstract}
The important Kemeny problem, which consists of computing median consensus rankings of an election with respect to the Kemeny voting rule, admits important applications in biology and computational social choice \cite{andrieu,social,3/4,brancotte} and was generalized recently via an interesting setwise approach by   Gilbert et. al. \cite{setwise-aaai, setwise}. Our first results establish  optimal quantitative  extensions of the Unanimity property and the well-known $3/4$-majority rule of Betzler et al. \cite{3/4} for the classical Kemeny median problem. Moreover, by elaborating an exhaustive list of quantified axiomatic properties (such as the Condorcet and Smith criteria, the $5/6$-majority rule, etc.) of the $3$-wise Kemeny rule where not only pairwise comparisons but also the discordance
between the winners of subsets of three candidates are also taken into account, we come to the conclusion that the $3$-wise Kemeny voting scheme induced by the $3$-wise Kendall-tau distance presents interesting advantages in comparison with the classical Kemeny rule. For example, it satisfies several  improved manipulation-proof properties. 
Since the $3$-wise Kemeny problem is NP-hard, our results also provide some of the first useful space reduction techniques by determining the relative orders of pairs of alternatives. Our works suggest similar interesting properties of higher setwise Kemeny voting schemes which justify and compensate for the more expensive computational cost than the classical Kemeny scheme. 
\end{abstract}

\section{Introduction} 
In this article, by an election we mean a finite collection $C=\{c_1, \dots, c_n\}$ of candidates (alternatives) together with a voting profile consisting of a finite number of votes which are not necessarily distinct. Here, a ranking or a vote is simply a complete and strict total ordering $\pi \colon c_{\pi(1)}> c_{\pi(2)}> \dots >c_{\pi(n)}$ which we identify with a permutation of $\{1, 2, \dots, n\}$ also denoted by $\pi$. The notation $x>y$ means that $x$ is ranked before $y$. Among several natural distances of the space of all rankings, the Kendall-tau distance, which is also the bubble-sort distance between two permutations, is one of the most prominent distances which counts the number of order disagreements between pairs of elements in two permutations. More generally, we have the more refined notion of $k$-wise Kendall-tau distance recently introduced in \cite{setwise} (see Definition~\eqref{e:definition-k-wise-distance}) which moreover takes into consideration the disagreement between the winners of subsets of at most $k$ candidates.  \par 
In the well-known Kemeny problem (cf. \cite{kemeny}, \cite{kemeny-snell}, \cite{young}), the objective is to determine the set of medians, i.e., permutations whose total distance to the voting profile is minimized. Hence, for the classical Kemeny rule which is induced by the  Kendall-tau distance, a median is simply a ranking 
that maximizes the number of pairwise agreements with the voting profile. One of the most important interpretations of the Kemeny rule is that it is a maximum likelihood estimator of the correct ranking (see \cite{kemeny}). 
\par 
Motivated by the NP-hardness of the various Kemeny problems (see \cite{dwork}, \cite{bachmeier}, \cite{setwise}), 
our main goal is to formulate new quantitative results concerning the majority rules in $k$-wise Kemeny voting schemes associated with the $k$-wise Kendall-tau distance  introduced recently in  \cite{setwise}, notably for $k=2,3$, which also provide some more refined space reduction techniques to the Kemeny problem than several existing techniques in the literature. Moreover, we establish several new properties of the $3$-wise Kemeny voting scheme such as a quantified version of the Condorcet criterion and more generally the Smith criterion and the Extended Condorcet criterion (Theorem~\ref{t:3-wise-condorcet}, Theorem~\ref{t:smith-criterion-main}, Theorem~\ref{t:XCC-3-wise-main}). 
Recent results in \cite{kien-sylvie-MOT}  show that the well-known 
 $3/4$-majority rule of Betzler et al. \cite{3/4} for the classical Kemeny rule, namely, the  $2$-wise Kemeny voting scheme, is only valid for small elections of no more than $5$ candidates with respect to the $3$-wise Kemeny scheme. In this paper, without restriction on the number of candidates, 
 we establish a $5/6$-majority rule (Theorem~\ref{l:5/6-majority-3-wise-main}, Theorem~\ref{l:5/6-majority-3-wise}) which serves as the $3$-wise counterpart of the $2$-wise $3/4$-majority rule.  
\par 
Our analysis provides  strong quantified evidence which shows that the $3$-wise Kemeny voting scheme is more suitable than the classical Kemeny voting scheme in many aspects. In particular, by  taking into account not only pairwise discordance but also the discordance between the winners of subsets of three candidates (see Definition~\eqref{e:definition-k-wise-distance} below), we show that the $3$-wise Kemeny voting scheme is  more resistant to coalitional manipulation than the classical Kemeny rule. More specifically, we prove that  it is much more difficult for a candidate to win an election or event to simply win another specific candidate in an election under the $3$-wise Kemeny voting scheme and almost all of the best-known space reduction techniques for the classical Kemeny rule fail for the $3$-wise Kemeny voting scheme (see, e.g., Table~\ref{table:summary}), including 
the powerful Major Order Theorems discovered in \cite{hamel-space} and the  Condorcet Criterion. 
For example, we show that even when a candidate wins the $2/3$ majority in every duel with all other candidates, this candidate may still lose the election according to the $3$-wise Kemeny voting scheme (see Theorem~\ref{c:lower-bound-2/3-theorem-5.1}). This phenomenon is rather surprising when compared to the Condorcet criterion for the classical Kemeny scheme where a Condorcet winner, namely a candidate who is preferred by more voters than any other candidate, must be the unique winner of the election. However, it turns out that a candidate obtaining a $3/4$ majority in every duel must be the unique winner in the $3$-wise Kemeny voting scheme (see Theorem~\ref{t:3-wise-condorcet}).  
\par

\subsection{The $k$-wise Kemeny rule} 

Let $k \geq 2$ be an integer and let $C$ be a finite set of candidates. Let $S(C)$ be the set of all rankings of $C$. Let $\Delta^k(C) \subset 2^C$ be the collection of all subsets of $C$ which contain no more than $k$ elements. To take into consideration disagreements not only on pairs of candidates but also on all subsets of candidates of cardinality at most $k$, the \emph{$k$-wise Kendall-tau distance} $ d^k_{KT}(\pi, \sigma)$ between two rankings $\pi, \sigma$ of $C$ is defined by (cf.~ \cite{setwise}): 
\begin{equation}
\label{e:definition-k-wise-distance}
    d^k_{KT}(\pi, \sigma) = \sum_{S \in \Delta^k(C)} \left( 1- \delta_{\mathrm{top}_S(\pi), \mathrm{top}_S(\sigma)}\right) 
\end{equation}
where $\mathrm{top}_S(\pi) \in S$ denotes the highest ranked element of the restriction $\pi\vert_S$ of $\pi$ to $S$ and $\delta_{x,y}$ denotes the Kronecker symbol which is equal to $1$ if $x=y$ and is equal to zero otherwise. 
The $k$-wise Kendall-tau distance between a ranking $\pi$ of $C$ and a collection of rankings $A$ of $C$ is defined as 
\begin{equation}
\label{e:3-wise-distance-to-a-set}
    d^k_{KT}(\pi,A) = \sum_{\sigma \in A}  d^k_{KT}(\pi,\sigma). 
\end{equation} 
\par 
Let $V$ be the voting profile of the election. Then we say that a ranking $\pi$ of $C$ is a \emph{median of the election with respect to the $k$-wise Kemeny rule} or simply a \emph{$k$-wise median} if 
\begin{equation*}
    d^k_{KT}(\pi,V) = \min_{\sigma \in S(C)} d^k_{KT} (\sigma, V). 
\end{equation*}
\par 
It is clear that for $k=2$, we recover the definition of the usual Kendall-tau distance  $d^2_{KT} = d_{KT}$. Moreover, 
it was shown in \cite{setwise} that the decision  variant of the $k$-wise Kemeny aggregation problem is NP-complete for every constant $k \geq 3$. 
\par 
We shall pay particular attention to the cases $k=2$ and $k=3$. Our results suggest that the $3$-wise Kemeny rule is more suitable than the $2$-wise Kemeny rule since it puts more weight on candidates who are more frequently ranked in top positions in the votes, which partially justifies the utility of the $3$-wise Kemeny voting scheme over the classical $2$-wise Kemeny scheme. Indeed, the classical Kemeny rule  puts equal weight on the head-to-head competition of two candidates $x, y$ regardless of the absolute positions of $x,y$ in each vote.  Nevertheless, typical voters in real-world settings only pay attention to a shortlist of their  favorite candidates and normally  put a somewhat arbitrary order for the rest of the candidates or simply do not indicate any preference for such candidates. 
As shown by an example given in \cite{kien-sylvie-MOT}, such voting   behavior creates undesirable noises which makes the $2$-wise Kemeny rule  vulnerable to manipulation that can  alter the consensus  ranking. Consequently, we should somehow reduce the weight of the  duel wins among unfavorite candidates of each vote. A possible solution for this problem is to use the $3$-wise Kemeny voting scheme as the weight of the duel $x$ vs $y$ in a ranking $x>A>y>B$ is multiplied by $|B|$ as we count $|B|$ subsets of the form $\{x,y,z\}$ where $z \in B$. Note also that the above explained imperfection of the Kemeny rule motivated the notion  of weighted Kendall tau distances introduced by Kumar and Vassilvitskii \cite{kumar} as well as the notion of setwise Kemeny distance of Gilbert et. al. \cite{setwise}. 
\par 
While the $3$-wise Kemeny voting scheme enjoys the above-mentioned desirable property as well as the Majority criterion (Lemma~\ref{l:majority-criterion-k-wise}), we prove that it suffers several minor  drawbacks compared to the classical Kemeny voting scheme. Notably, the 3-wise Kemeny  scheme does not satisfy the Reversal symmetry (Lemma~\ref{l:reversal-symmetry-3-wise}), the Condorcet criterion  as well as the Condorcet loser criterion (Lemma~\ref{l:condorcet-loser-3-wise}). 

\subsection{The classical Kemeny voting scheme revisited}
By applying our analysis technique to the case of the $2$-wise Kemeny voting scheme, i.e., the  classical Kemeny rule, we can strengthen various reduction techniques such as the Always theorem (which states that if a candidate is always preferred over another candidate then the same holds for every Kemeny median),  and the $3/4$-majority rule of Betlzer et al. \cite{3/4}.  Thus, our results extend the range of applications of the Condorcet method (see Theorem~\ref{l:extended-majority} and Theorem~\ref{t:extended-always}). 
Since the $3/4$-majority rule and the Always theorem are particularly useful for data in real-life competitions where the orders of the candidates do not vary much among the votes, our results should find a wide range of meaningful applications in practice.

\subsection{Summary} 
Table~\ref{table:summary} gives a summary of our main results as well as the comparison of various space reduction criteria for the classical and $3$-wise Kemeny voting schemes.   
Table~\ref{table:applicability-AT} presents the mean percentage of cases (over a total of 100 000 uniformly generated instances) to which the Always Theorem, the $2$-wise and $3$-wise Extended Theorems (Theorem~\ref{t:extended-always}, Theorem~\ref{t:3-wise-unanimity-general}) are applicable.  
 
\begin{table}[h!]
\centering
\caption{Space reduction results for setwise Kemeny schemes}
\label{table:summary}

\begin{tabularx}{1\textwidth}{ 
  | >{\centering\arraybackslash}X 
  | >{\centering\arraybackslash}X 
  | >{\centering\arraybackslash}X |}
\hline
\textbf{Criterion} & \textbf{Kemeny rule} & \textbf{$3$-wise Kemeny rule} \\ 
\hline  
Monotonicity & Yes & Gilbert et al. 2022 \cite{setwise}
\\ 
\hline  
 Condorcet loser criterion  & Yes & No, Lemma~\ref{l:condorcet-loser-3-wise} \\   
 \hline
 Reversal symmetry & Yes & No, Lemma~\ref{l:reversal-symmetry-3-wise} \\   
 \hline
Majority criterion &   Lemma~\ref{l:majority-criterion-k-wise} &   Lemma~\ref{l:majority-criterion-k-wise}\\ 
\hline
 Extended Always theorem (Pareto efficiency, unanimity) & Theorem~\ref{t:extended-always} & Theorem~\ref{t:3-wise-unanimity-general}  \\ \hline
Condorcet criterion
 &  Condorcet 1785 \cite{condorcet}   &  No,  Theorem~\ref{c:lower-bound-2/3-theorem-5.1} \\\hline
 $2/3$-Condorcet criterion 
 &  Yes   & No,  Theorem~\ref{c:lower-bound-2/3-theorem-5.1}   \\\hline
$3/4$-Condorcet criterion 
 &  Yes   & Theorem~\ref{t:3-wise-condorcet}  \\\hline
  Smith criterion & Smith 1973 \cite{smith} & No, Theorem~\ref{c:lower-bound-2/3-theorem-5.1}   \\\hline
   $2/3$-Smith criterion & Yes & No,   Theorem~\ref{c:lower-bound-2/3-theorem-5.1}   \\\hline
 $3/4$-Smith criterion & Yes  & Theorem~\ref{t:smith-criterion-main}   \\ \hline 
$3/4$-Smith-IIA & Yes, even Smith-IIA & No, Example~ \ref{ex:3/4-smith-iia}
 \\ \hline
Extended Condorcet criterion 
 &   Truchon 1990 \cite{truchon-XCC}   & No,  Theorem~\ref{c:lower-bound-2/3-theorem-5.1}  \\\hline
$3/4$-Extended Condorcet criterion 
 &  Yes & Theorem~\ref{t:XCC-3-wise-main}  \\\hline
$3/4$-majority rule   & Betzler et al. 2014 \cite{3/4}  & Valid only for elections of $5$ candidates or less, Phung-Hamel \cite{kien-sylvie-MOT}\\\hline
Extended $s$-majority rule  &    Theorem~\ref{l:extended-majority}, Corollary~\ref{c:extended-majority} & Yes for  $s\geq 5/6$, see below 
\\\hline
$5/6$-majority rule & Yes & Theorem~\ref{l:5/6-majority-3-wise-main}, Theorem~\ref{l:5/6-majority-3-wise}, Corollary~\ref{c:5/6-majority-rule-consequence}   \\ 
\hline
Major Order Theorems & Milosz-Hamel 2020 \cite{hamel-space}  & No, see Phung-Hamel \cite{kien-sylvie-MOT} \\ 
\hline 
$3$-wise Major Order Theorems & & Phung-Hamel \cite{kien-sylvie-MOT}. 
 \\ 
\hline
\end{tabularx}

\end{table}

\section{Preliminaries}

\begin{definition}
Let $x,y$ be candidates in an election with voting profile $V$ and let $s\in [0,1]$. When $x$ is ranked before $y$ in a vote $ v \in V$, we denote $x>^v y$ or simply $x>y$ when there is no possible confusion.  
We write $x \geq_s y$, resp. $x>_sy$, if  $x>y$ in at least, resp. in more than, $s|V|$ votes. 
\end{definition}

The next lemma is a simple but useful observation which serves as a weak form of the transitivity property in the context of voting.

\begin{lemma}[see \cite{kien-sylvie-MOT}]
\label{l:transitive}
    Let $x,y,z$ be three distinct  candidates in an election with voting profile $V$. Suppose that $x\geq_s y$ and $x\geq_s z$ for some $s \in [0,1]$. Then in at least $(2s-1)|V|$ votes, we have $x >y, z$, i.e., $x>y$ and $x>z$.  
\end{lemma}

\begin{proof}
Let $A, B \subset V$ be the sets of votes in which $x>y$ and $x>z$ respectively. Since $x\geq_s y$ and $x\geq_s z$, we deduce that $|A|\geq s|V|$ and $|B|\geq s |V|$. Consequently, from the set inequality $|A \cup B| \leq |A|  + |B| - |A \cap B|$, we obtain the following estimation 
\begin{align*}
    |A\cap B| & \geq |A| + |B| - |A \cup B|   \geq s|V|+ s|V| - |V| = (2s-1) |V|.
\end{align*}
\par 
Since in every vote $v \in A \cap B$, we have $x>y,z$, the conclusion thus follows. 
\end{proof}

 \section{An optimal extension of the Always theorem for the Kemeny rule} 
 
We have the following optimal extended version of the Always theorem \cite{always} (Unanimity property) for the Kemeny rule. The Always theorem states that if a candidate $x$ is ranked before another candidate $y$ in every vote then $x$ must be ranked before $y$ in every $2$-wise median of the election. In essence, our result implies that $x$ does not need to win $y$ in every single vote in order to win $y$ in $2$-wise median rankings.  

\begin{theorem}
\label{t:extended-always}
Let $x, y$ be candidates in an election with $n \geq 2$ candidates. Suppose that $x \geq_\alpha y$ for some $ 1-\frac{1}{n} < \alpha \leq 1$. Then we have $x >y$ in every Kemeny ranking of the election. 
\end{theorem}

For example, we can choose   $\alpha=0.917$ when $n=12$ and any  $\alpha>0.95$ when $n=20$:
\begin{center}
\begin{tabular}{ | M{2.5em} |  M{1em} | M{1cm}  | M{1cm}  | M{1cm} | M{1cm} | M{1cm}| M{1cm}| M{1cm}| M{1cm}| } 
  \hline
  $n$ & 2 & 3 & 4 & 5 & 6 & 8 & 10 & 12 & 20 \\ 
  \hline
  $1-\frac{1}{n}$ & 0.5 & 0.667 & 0.75 & 0.8 & 0.833  &   0.875 &   0.9 &   0.917 & 0.95 \\ 
  \hline
\end{tabular}
\end{center}
\par 

Note that our result is optimal in the following sense: for every even $n \geq 2$ and every $0 \leq s \leq 1-\frac{2}{n}$, there exists by \cite[Proposition 1]{hamel-space} an election with $n$ candidates together with a pair of candidates $(x,y)$ such that $x\geq_s y$ but $x<y$ in every Kemeny ranking of the election. 

\begin{proof}
Suppose on the contrary that there exists a Kemeny ranking $r$ of the election in which $x <y$. Let $K$ be the set of all candidates ranked between $x$ and $y$ in the ranking $r$. Then it follows that $y > K >x $. Let $k = \vert K \vert$ then it is clear that $k \leq n-2$. 
\par 
Let $V$ be the multiset of all votes and $m= \vert V \vert$. For every $z \in K$, we denote:
\begin{equation*}
    x_z=\vert  \{ v \in V \colon x >^v z \}\vert, \quad y_z =  \vert  \{ v \in V \colon z >^v y \}\vert.
\end{equation*}
\par 
We claim that $x_z + y_z \geq \alpha m$ for all $z \in K$. Indeed, observe that $m - x_z$ is the number of votes in which $z>x$. Among these votes, there are at most $(1-\alpha)m$ votes in which $x<y$ 
since $x \geq_\alpha y$ by hypothesis. It follows from the transitivity in a ranking that $z>y$ in at least $m-x_z - (1-\alpha)m$ votes. Consequently, we find that 
\[
y_z \geq m-x_z - (1-\alpha)m = \alpha m - x_z
\]
and the claim is thus proved. 
\par 
Consider the ranking $r^*$ which is obtained from $r$ by replacing the block $y >K>x$ by the block $x>K>y$ where the order of the candidates in $K$ is not modified. Let $\Delta = d_{KT}(r^*, V)- d_{KT}(r,V)$ then it is clear from the definition of the Kendall-tau distance that only the pairs $(x,y)$, $(x,K)$, $(y,K)$ can contribute to the Kemeny  score difference $\Delta$.  
\par 
Since $x\geq_\alpha y$, the pair $(x,y)$ contributes at most $(1- 2 \alpha)m$ to $\Delta$. Let $z \in K$, then the contribution to $\Delta$ of the pair $(x,z)$ is at most $(m-x_z)-x_z=m-2x_z$. 
Similarly, for every $z \in K$, the pair $(y,z)$ contributes at most $m-2y_z$ to $\Delta$. 
\par 
Therefore, the score difference $\Delta$ is bounded by: 
\begin{align*}
\label{e:t:extended-always}
    \Delta & \leq (1-2\alpha)m + \sum_{z \in K} ((m-2x_z) + (m -2y_z)) \\ 
    & = (1-2\alpha) m + 2 \sum_{z \in K} (m - (x_z +y_z))
    \\
    & \leq (1-2\alpha) m + 2 \sum_{z \in K} (m - \alpha m) 
    \\
    & = m (1-2\alpha + k(1-\alpha))
    \\
    & \leq m (1-2\alpha + (n-2)(1-\alpha)) \numberthis{}
\end{align*} 
where we use the inequality $x_z +y_z \geq \alpha m$ for every $z \in K$ in the second inequality and the fact that $1-\alpha >0$ and $k \leq n-2$ in the last inequality. 
\par 
Now observe that: 
\begin{align*}
    \Delta <0 & \Longleftarrow 1-2\alpha + (n-2)(1-\alpha) <0 \\
    & \Longleftrightarrow  \frac{n-1}{n} < \alpha.  
\end{align*}
\par 
Hence, $\Delta <0$ for such $\alpha$ and we deduce that in this case, $r$ cannot be a Kemeny ranking of the election. We thus obtain a contradiction and the proof is complete. 
\end{proof}

\section{An optimal extension of the 3/4-majority rule for the Kemeny rule}
  
\par 
Following \cite{3/4}, a \emph{non-dirty pair of candidates} in an election with respect to a threshold $s \in [0, 1]$ is a pair $(x, y)$ such that either $x$ is ranked before $y$ in a proportion of at least $100s\%$ of the votes, or $y$ ranked before $x$ in a proportion of at least $100s\%$ of the votes. Then we say that a candidate is \emph{non-dirty} if $(x,y)$ is a non-dirty pair with respect to the threshold $s$ for every other candidate $y \neq x$. An election with a certain voting rule satisfies the  \emph{$s$-majority rule} if for every non-dirty candidate $x$ with respect to the threshold $s$ and every other candidate $y \neq x$, the relative positions of the pair $(x,y)$ in every median is determined by the $s$-majority head-to-head.  
\par 
By well-known results in \cite{3/4}, we know that every election satisfies the $3/4$-majority rule. Our optimal extension of the $3/4$-majority rule (see Theorem~\ref{l:extended-majority}, Corollary~\ref{c:extended-majority}) implies  that the $s$-majority rule holds for $s=0.725$, resp. $s= 0.74$, resp. $s=0.744$ for elections consisting of at most $6$, resp. $12$, resp. $20$, candidates.
\par 
Our extension of the $3/4$-majority rule for the Kemeny voting scheme is the following. Note that our criterion is asymptotically optimal (see Corollary~\ref{c:extended-majority}, Remark~\ref{remark:l(s)-betzler},  Remark~\ref{r:extension-3/4-asymptotic-optimal}, and Table~\ref{table:extension-3/4-optimal}).  

\begin{theorem}
\label{l:extended-majority}
Let $s \in \left[\frac{1}{2}, \frac{3}{4}\right]$ and let $q= \frac{3}{2} - s + \varepsilon$. 
Let $x$ be a non-dirty candidate in an election with at most $n\geq 3$ candidates with respect to the $s$-majority rule and suppose that:
\begin{equation*}
    n \leq n(s)=  \frac{\sqrt{(1-s)(7-9s)}+4-5s}{3-4s} ; \quad \varepsilon > \frac{s-1}{n-1}. 
\end{equation*}
\par 
Then for every candidate $y$ such that $x \geq_q y$, resp. $y \geq_q x$, we have $x > y$, resp. $y>x$,  in every Kemeny ranking of the election.
\end{theorem}

For example,  when $n=6$, resp. $n=12$, resp. $n=20$,  we can take $(s,q)=(0.716, 728)$, resp. $(s,q)=(0.737,0.74)$, resp. $(s,q)=(0.743, 0.744)$. In particular, we deduce from Theorem~\ref{l:extended-majority} that the $s$-majority rule holds (with respect to the Kemeny voting scheme) for $s=0.728$, resp. $s= 0.74$, resp. $s=0.744$,  when $n=6$, resp. $n=12$, resp. $n=20$. Hence, the result will find some meaningful applications for elections with a small number of candidates but with a large number of votes, e.g., presidential elections.  

\begin{proof}
Let $C$ be the set of all candidates in the election then $n=\vert C \vert $ and let $V$ be the set of all votes. 
Suppose on the contrary that there exists a Kemeny ranking $r$ together with a candidate $y$ and a non-dirty candidate $x$ with respect to the $s$-majority rule such that $x \geq_q y$ but $y>x$. 
\par 
Let $K$ be the set of all  candidates ranked between $x$ and $y$ in $r$. It follows that $y > K >x$ in $r$. Since $x$ is non-dirty with ratio $s$, we have a partition $K= L \coprod R$ where  
\begin{equation*}
    L = \{z \colon z \geq_s x\}, \quad R = \{z \colon z <_s x\}.  
\end{equation*}
\par 
We are going to show in the sequel that the modified ranking $r^*$ of $r$ with $L > x > y >R$ has lower Kemeny score than $r$. Indeed, the only pairs of candidates which can contribute to the score difference $\Delta = d_{KT}(r^*, V) - d_{KT}(r, V)$ are $(x,y)$, $(x, R)$, $(y, L)$, and $(L,R)$ where e.g. $(x,R)$ means any pair $(x,z)$ where $z \in R$ and similarly for $(y, L)$ and $(L,R)$. 
\par 
Since $x \geq_q y$, the pair $(x,y)$ contributes at most $\vert V \vert - q\vert V \vert$ to $d_{KT}(r^*, V)$ but contributes at least $ q \vert V \vert$ to $d_{KT}(r, V)$. Hence, the contribution of $(x,y)$ to $\Delta$ is at most
\[
\vert V \vert - q\vert V \vert -  q \vert V \vert = (1-2q)\vert V \vert.
\]
\par 
Similarly, the contribution of the  $\vert R \vert $ pairs $(x,L)$ to the difference $\Delta$ is at most $(1-2s)\vert V\vert \vert R \vert$. 
\par 
For every $z \in L$ and $t \in R$, we infer from the definition of $R$ and $L$ that $z\geq_s x$ and $x\geq_s t$. Therefore, $z > x >t$ and thus $z>t$ in at least $(2s-1)\vert V \vert$ votes. It follows that the $\vert R \vert \vert L\vert$ pairs $(R,L)$ contribute at most 
\[
(1-2(2s-1))\vert V\vert \vert R \vert \vert L \vert = (3-4s)\vert V\vert \vert R \vert \vert L \vert. \] 
\par 
Similarly, for every $z \in L$, we find that $z > x >y$ and thus $z>y$ in at least 
$(q+s -1)\vert V \vert$ votes. Consequently, the $\vert L \vert$ pairs $(L,y)$ contribute at most 
\[
(1-2(q+s-1))\vert V\vert  \vert L \vert = (3-2(q+s))\vert V\vert  \vert L \vert \]
to the score difference $\Delta$. 
\par 
Let us denote $v = \vert V \vert$, $k=\vert K \vert$, and $a= \vert R\vert$ then  $\vert L \vert = k-a$. To summarize, we deduce by substituting $q= \frac{3}{2}-s+\varepsilon$ from the above inequalities that $\Delta$ is bounded as follows:  
\begin{align}
\label{e:l:extended-majority-main}
    \frac{\Delta}{v} & \leq (1-2q) + (1-2s)a + (3-4s)a(k-a) + (3-2(q+s))(k-a) \\
    & = (2s-2-2\varepsilon) + (1-2s)a+ (3-4s)a(k-a)  -2\varepsilon (k-a) \nonumber\\
    & = F(a) \nonumber
\end{align}
where $F(a)=(4s-3)a^2+ (1-2s+2\varepsilon +(3-4s)k)a + 2s-2-2\varepsilon (k+1)$. 
\par 
Observe that for fixed $\varepsilon > \frac{s-1}{n-1}$, we have $ 
F(0) = 2s-2- 2\varepsilon (k+1) <0$ since 
$\frac{s-1}{n-1}\geq \frac{s-1}{k+1}$ as $s-1<0$ and $0<k +1 \leq n-1$. 
On the other hand, the critical point of $F$ is 
\[
a_0= \frac{1-2s+2\varepsilon +(3-4s)k}{2(3-4s)}. 
\] 
\par 
Note that since $3-4s>0$, we have $a_0 \leq 0$ if and only if $k \leq \frac{2s -1 -2\varepsilon}{3-4s}$. The latter condition is always satisfied if 
\begin{equation}
\label{e:3/4-majority} 
n \leq \frac{2s -1 -2\varepsilon}{3-4s} +2 \leq 
\frac{2s -1 + \frac{2-2s}{n-1}}{3-4s} +2 
\end{equation} 
since $k \leq \vert V \vert -2 = n-2$ and $\varepsilon > \frac{s-1}{n-1}$. As the leading coefficient of the quadratic function $F$ is $ 4s-3 <0$ and $a \geq 0$, we have for all such $n$ that $F(a) \leq F(0) < 0$. Consequently, $ \Delta \leq v F(a) <0$ and we conclude that $d_{KT}(r^*, V) < d_{KT}(r, V)$. This contradiction shows that $r$ cannot be a Kemeny ranking of the election. 
In other words,  if $x \geq_q y$ then we must have $x > y$ in every Kemeny ranking whenever \eqref{e:3/4-majority} is satisfied. The case $x \leq_q y$ is completely similar. 
\par 
We will now solve the condition $n  \leq 
\frac{2s -1 + \frac{2-2s}{n-1}}{2(3-4s)} +2$ in \eqref{e:3/4-majority}. By setting $u=n-1 >0$, a straightforward algebra then shows that the inequality is equivalent to: 
\[
f(u)=(3-4s)u^2 - (2-2s)u - 2+2s \leq 0. 
\]
\par 
Observe that $f(0)=-2+2s <0$. Hence, using again the fact that the leading coefficient of the quadratic function $f$ is $3-4s >0$, 
we obtain the following description of all the solutions $n$ of \eqref{e:3/4-majority}: 
\[
n = u+1 \leq n(s)= \frac{\sqrt{(1-s)(7-9s)}+4-5s}{3-4s}
\]
and the proof is complete. 
\end{proof}   
\par 
Consequently, the Extended $s$-majority rule Theorem~\ref{l:extended-majority} shows essentially that the $s$-majority rule still holds if $s \in \left[ 
 \frac{7}{10} , \frac{3}{4}\right[$ under the condition that we  bound the number of the candidates. 
\begin{remark}
\label{remark:l(s)-betzler}
For every $s \in \left]\frac{2}{3}, \frac{3}{4}\right[$, the counterexample to the $s$-majority rule  given in \cite[Proposition 1]{3/4} requires $\lceil l \rceil +2 $ candidates where 
\begin{equation}
\label{e:l(s)-betzler}
l= l(s)= \frac{3s}{3-4s} 
\end{equation} which is linear in $(1- 4s/3)^{-1}$. Note that  $l(0.72) = 18$. 
\end{remark}
 
\begin{corollary}
\label{c:extended-majority}
Let $t\in \left[\frac{1}{2}, \frac{3}{4}\right[$ and let us define:
\begin{equation}
\label{e:n(t)-betzler}
n(t)=  \frac{\sqrt{(1-t)(7-9t)}+4-5t}{3-4t}, \quad s(t, n(t)) = \frac{3}{2}-t - \frac{1-t}{n(t)-1}.
\end{equation}
\par 
Then for any $s<s(t,n(t))$, the $s$-majority rule holds for all elections with $n$   candidates such that $n \leq n(t)$.  
\end{corollary}
\begin{proof}
By the proof of Theorem~\ref{l:extended-majority} where
\begin{romanenumerate}
\item $\varepsilon = \frac{t-1}{n(t)-1}$, \, $q=s(t,n(t))$, 
\item 
$t$ in our corollary plays the role of $s$ in Theorem~\ref{l:extended-majority},
\end{romanenumerate} 
\par 
we deduce that for all 
\[
s<\max (t,q) = \max(t,s(t,n(t))),
\] 
the $s$-majority rule holds for all elections with $n \leq n(t)$  candidates. Note that for $\varepsilon=\frac{t-1}{n(t)-1}$ but $s< \max(t,s(t,n(t)))$, we still have $F(0)<0$ in the proof of Theorem~\ref{l:extended-majority}. 
\par 
On the other hand, a straightforward analysis of the function $s(t,n(t))$ shows that $s(t,n(t))\geq t$ for all $t \in \left[\frac{1}{2}, \frac{3}{4}\right[$. The conclusion thus follows. 
\end{proof} 
\par 

\begin{remark} 
\label{r:extension-3/4-asymptotic-optimal} 
A direct computation shows that for every $t \in \left[\frac{1}{2}, \frac{3}{4}\right[$, the functions $l(s) $, $n(t)$, and $s(t,n(t))$  (cf. \eqref{e:l(s)-betzler} and \eqref{e:n(t)-betzler}) satisfy: 
$$
4 < \frac{l(s(t,n(t))+2}{n(t)} < \frac{9}{2}, \quad \lim_{t \to \frac{3}{4}} s(t,n(t) = \frac{3}{4}, \quad \text{and} \quad \lim_{t \to \frac{3}{4}}  \frac{l(s(t,n(t))+2}{n(t)} = \frac{9}{2}. 
$$
\par 
Consequently, by Remark~\ref{remark:l(s)-betzler} and Corollary~\ref{c:extended-majority}, 
we can conclude that  our extension of the $s$-majority rule Corollary~\ref{c:extended-majority} is indeed asymptotically optimal for the Kemeny rule. The following table \ref{table:extension-3/4-optimal} provides the first few values of the quadruples  $(t,s,n,l+2)$ where 
\[
n=\lfloor n(t) \rfloor, \quad s=s(t,n(t)), \quad  l=\lceil l(s(t,n(t))) \rceil. 
\]
\begin{table}[h!]
\centering 
\caption{Optimality of the Extended $s$-majority
rule}

\begin{tabularx}{1\textwidth}{  
  | >{\centering\arraybackslash}X 
   | >{\centering\arraybackslash}X 
    | >{\centering\arraybackslash}X 
     | >{\centering\arraybackslash}X 
     | >{\centering\arraybackslash}X 
  | >{\centering\arraybackslash}X 
  | >{\centering\arraybackslash}X |} 
  \hline
   $\textbf{t}$  & 0.5 & 0.6  & 0.7 & 0.71 & 0.72 & 0.73  \\ 
   \hline
     $\textbf{s}$ & 0.691  &  0.7  &   0.721 & 0.729 &  0.725 & 0.735  \\ 
     \hline 
     $\textbf{n}$ & 2 &  3  &   4 & 5 & 6& 8  \\ 
     \hline 
      $\textbf{l}+2$ & 11  &  13  &  21 & 24 & 29  & 38 \\ 
  \hline \hline 
 $\textbf{t}$ &  0.74 & 0.742 & 0.744 &0.746 & 0.748 & 0.749  \\ 
  \hline
  $\textbf{s}$  & 0.741 &  0.743 & 0.745 & 0.746 & 0.748 & 0.749 \\ 
  \hline
  $\textbf{n}$ & 14 &  18 & 23 & 33 & 64 & 127\\ 
  \hline
  $\textbf{l}+2$ & 67 & 81 & 104 &151 & 292 & 573  \\ 
  \hline
\end{tabularx}
\label{table:extension-3/4-optimal}
\end{table}
\end{remark}
\par

\section{$3$-wise Condorcet consistency and $3$-wise Smith criterion} 

\subsection{Condorcet criterion} 
While the classical Kemeny voting scheme satisfies the Condorcet criterion, the $k$-wise Kemeny voting scheme is not a Condorcet scheme whenever  $k \geq 3$ since a \emph{Condorcet winner} might not be the winner in median  rankings (see \cite[Proposition 3]{setwise}). Here, a candidate $x$ is a Condorcet winner in an election if $x \geq_{1/2} y$ for all other candidate $y \neq x$. 
However, we establish the following result which shows that if a candidate wins by a large enough margin (slightly less than $75\%$) in every duel, then it is necessarily the \emph{unique winner} in every Kemeny ranking with respect to the 3-wise Kemeny rule. 

\begin{theorem}
\label{t:3-wise-condorcet}
Let $x$ be a candidate in an election with $n \geq 2$ candidates. Suppose that  $x \geq_\alpha y$ for all candidate $y \neq x$ where 
\[
\alpha > f(n)=  \frac{3n-5}{4n-6}.
\]
\par Then $x$ is the winner in every median of the election with respect to the $3$-wise Kemeny rule. In particular, if $x \geq_{3/4} y$ for all candidate $y \neq x$ then $x$ is the winner in every 3-wise median. 
\end{theorem}
\par 
For example, the following table gives us several values of $f(n)$: 
\begin{center}
\begin{tabular}{ | M{3em} | M{1cm}  | M{1cm}  | M{1cm} | M{1cm} | M{1cm}| M{1cm}| M{1cm}| } 
  \hline
  $n$  & 4 & 8 & 12 & 20 & 30 & 50 \\ 
  \hline
  $f(n)$ & 0.7 & 0.731 &  0.738 &   0.743 &   0.746 & 0.747 \\ 
  \hline
\end{tabular}
\end{center}

\begin{proof}
Let $C$ be the set of all candidates then $\vert C \vert = n$. 
Suppose on the contrary that there exists a Kemeny ranking $r$ of the election with respect to the distance $d^3_{KT}$ in which $x$ is not the winner. Let $y \neq x$ be the candidate which is ranked immediately before $x$ and let $L, R$ be respectively the ordered sets of all candidates ranked before $y$ and after $x$ in the ranking $r$. In other words, the ranking $r$ can be written as  $L > y  >x >R$. 
\par 
We consider the ranking $r^*$ obtained from $r$ by simply exchanging the positions of the candidates $x$ and $y$ while keeping the positions  of all other candidates. Hence, 
the ranking $r^*$ is $L > x >y >R$. 
We are going to show that $\Delta = d^3_{KT}(r^*, V)- d^3_{KT}(r,V) <0$ to obtain a contradiction. 
\par 
When restricted to $\Delta^2(C)$, i.e., to subsets of $C$ of cardinality at most $2$, we infer from the definition of $r^*$ and of the distance $d^3_{KT}$ that only the subset $\{ x,y\}$  can contribute to the score difference $\Delta$. Since $x \geq_\alpha y$, we find that $x >y$ in at least $\alpha v$ votes. Therefore, it is clear that the subset $\{ x,y\}$ contributes at most 
\begin{equation}
     \label{e:3-wise-3} 
(v- \alpha v) - \alpha v = (1-2\alpha) v. 
\end{equation}
to the score difference $\Delta$. 
\par 
Similarly, if we restrict to subsets of $C$ of cardinality $3$ then it is clear from the definition of $r^*$ that only the subsets of the form $S= \{x,y,z\}$, where $z\in C \setminus \{x,y\}$, can contribute to $\Delta$. Let us fix 
$z\in C \setminus \{x,y\}$. If $z \in L$ then 
$\mathrm{top}_S(r^*) =\mathrm{top}_S(r)=z$ and the contribution of $S$ to $\Delta$ is thus zero. 
\par 
Suppose that $z \in R$. 
Since $x \geq_\alpha y$ and $x \geq_\alpha z$, the number of votes in which we have $x >y$ and $x >z$ is at least $\alpha v - (v- \alpha v)= (2\alpha -1) v$. Let us fix  such a vote $\pi$.  Then observe that  $\mathrm{top}_S(r^*)=\mathrm{top}_S(\pi)=x$ while $\mathrm{top}_S(r)=y$ and the contribution of $S$ to $\Delta$ when computed with $\pi$ is -1.
\par 
Consequently, the total contribution to $\Delta$ of the subsets $S=\{x,y,z\}$, where $z \in   C \setminus \{x,y\}$, is at most 
\begin{equation} 
 \label{e:3-wise-4}
 \vert R \vert (v- 2(2\alpha -1)v) = \vert R \vert (3-4\alpha)v.
\end{equation} 
\par 
To summarize, we obtain from \eqref{e:3-wise-3} and \eqref{e:3-wise-4}   the following estimation on the 3-wise  score difference $\Delta$: 
\begin{align*}
\Delta & \leq (1-2\alpha) v + \vert R \vert  (3-4\alpha) v
\\
 & = v \left( 
3\vert R \vert +1 - 
 \alpha (4\vert R \vert +2) \right). 
\end{align*} 
\par 
Since $v >0$, we find that:
\begin{align*}
\label{e:3-wise-2}
    \Delta <0 & \Longleftarrow 3\vert R \vert +1 - 
 \alpha (4\vert R \vert +2) <0 \\
 & \Longleftrightarrow 
 \alpha > \frac{3 \vert R \vert +1}{4\vert R \vert +2} = \frac{3}{4}- \frac{1}{8 \vert R \vert +4 }. \numberthis{}  
\end{align*}
\par 
Since $\vert R \vert \leq n-2$, the inequality \eqref{e:3-wise-2} is always satisfied if 
\begin{equation*}
    \alpha > \frac{3}{4}- \frac{1}{8 (n-2) +4 }  = f(n). 
\end{equation*} 
\par 
Therefore, 
whenever $\alpha > f(n)$, we have $\Delta <0$ which implies that $r$ cannot be a Kemeny ranking with respect to the 3-wise rule. We conclude that $x$ must be the winner in every 3-wise  median and the proof is complete.  
\end{proof}

\subsection{Setwise Majority criterion and application} 
A weaker property than the Condorcet criterion is the \emph{Majority criterion} which says that if one candidate is ranked first by a majority (more than $50\%$) of votes, then that candidate must win the election. The following simple observation shows that the $k$-wise Kemeny voting scheme satisfies the Majority criterion for every $k \geq 2$. 

\begin{lemma}
\label{l:majority-criterion-k-wise}
Let $x$ be a candidate in an election such that $x$ wins more than $50\%$ of the votes. Then $x$ is the winner in every $k$-wise median ($k \geq 2$).  
\end{lemma}

\begin{proof}
Let $V$ be the voting profile. Suppose on the contrary that there exists a $k$-wise median $\pi \colon L >y>x>R$ where $y$ is a candidate and $L, R $ are (possibly empty) ordered sets of candidates. Consider the ranking $\pi^*\colon L >x>y >R$ and let $S$ be a subset of candidates. Then $S$ can contribute to the score difference $\Delta = d^k_{KT}(\pi^*, V) - d^k_{KT}(\pi, V)$ only if $x,y \in S$ and $L \cap S= \varnothing$. Note that since $k \geq 2$, such a subset $S$ always exists.
Then we have $x>S\cup\{y\}$ in more than  $|V|/2$ votes and thus $y>S\cup \{x\}$ in less than $|V|/2$ votes. Consequently, every subset of candidates $S$  contributes to $\Delta$ by an amount strictly less than  
\[
|V|/2 - |V|/2 = 0. 
\]
Hence, by summation over all subsets $S$ of size at most $k$, we deduce that $\Delta < 0$, which contradicts the hypothesis that $\pi$ is a $k$-wise median.  
\end{proof}

We obtain the following  application. 

\begin{corollary}
    \label{c:majority-criterion-k-wise}
    Let $\pi$ be a vote that appears in more than $50\%$ of the total votes in an election. Then $\pi$ is the unique median ranking of the election with respect to every $k$-wise Kemeny voting scheme for all $k \geq 2$. 
\end{corollary}

\begin{proof}
 Fix an integer $k \geq 2$ and a $k$ -wise median $\sigma$ of the election with voting profile $V$. Let us write $\pi \colon x_1 >x_2> \dots >x_n$ where $C=\{x_1, \dots, x_n\}$ is the set of all candidates and let $V_p$ be the election where we eliminate the candidates $x_1, \dots, x_p$ from the list of candidates and from all the votes. 
\par 
Let $\pi_1$, $\sigma_1$ be the induced rankings where we eliminate the candidate $x_1$ from $\pi$ and from $\sigma$ respectively. 
By Lemma~\ref{l:majority-criterion-k-wise}, the candidate $x_1$ is the unique  winner in every median thus in the ranking $\sigma$ in particular. 
Therefore, it is clear from the definition of $d^k_{KT}$ that  every subset $S \in \Delta^k(C)$ containing $x_1$ contributes a constant amount to $d^k_{KT}(r, V)$ for every $k$-wise median $r$ of $V$. We deduce that $\sigma$ is a $k$-wise median of the election $V$ if and only if $\sigma_1$ is a $k$-wise median in the election $V_1$. 
\par 
 Observe that $\pi_1$ occurs in more than $50\%$ of the votes in the election $V_1$ since $\pi$ appears in more than $50\%$ of the votes in the original election $V$. 
 \par 
By repeating the above argument, we infer again from Lemma~\ref{l:majority-criterion-k-wise} that $x_2$, the winner of $\pi_1$, must be the winner of $\sigma_1$, etc. By induction, we obtain  $\sigma=\pi$ and the proof is thus complete. 
\end{proof}

From the above result, we obtain counter-intuitive and extreme situations where in an election, a candidate may be the loser in every $k$-wise  median despite winning almost half of the votes. More surprisingly, this candidate may even lose to another candidate who consistently occupies the last two positions in every vote.  

\begin{example}
\label{ex:l:majority-criterion-k-wise}
Let $m \geq 1$ and $n\geq 0$. Let $V$ be the following voting profile in an election with $n+3$ candidates: 
\begin{align*}
 \pi &: w>z_1>\dots >z_n>x>y & (m+1 \text{ votes})\\
\sigma_1 & : y>A_1>x>w & (1  \text{ vote}) \\ 
    \dots &  &
    \\
\sigma_m & : y>A_m>x>w  &(1 
 \text{ vote}). 
\end{align*}
\par
Here, $A_1$, $\dots$, $A_m$ are arbitrary rankings of $\{z_1, \dots. z_n\}$. 
Then by Corollary~\ref{c:majority-criterion-k-wise}, the ranking $\pi$ is the unique $k$-wise median of the election for every $k \geq 2$. In particular, $y$ loses against $x$ in every median despite the fact that $y$ wins nearly half of the votes while $x$ always finishes among the last two.  
\end{example}

\par 

A generalization and stronger version of the majority criterion is the \emph{mutual 
majority criterion} which says that if there exists a partition $C= I \cup J$ of the set of candidates such that in more than half of the votes, we have $x>y$ for all $x \in I$ and $y \in J$, then the winner of the election must come from the set $I$. It is known that the $2$-wise Kemeny voting scheme satisfies the mutual majority criterion.

\subsection{$3$-wise Condorcet loser criterion and Reversal symmetry}

We know that the $2$-wise Kemeny voting scheme satisfies the Condorcet criterion as well as the \emph{Condorcet loser criterion}. The latter means that if a candidate loses every duel in an election then the candidate cannot be the winner.

 \begin{lemma}
     \label{l:condorcet-loser-3-wise} 
     The $3$-wise Kemeny scheme does not satisfy the Condorcet loser criterion. 
 \end{lemma}

 \begin{proof}
Let $C= \{x,y,z, t\}$ be a set of 4 candidates and consider the voting profile $V$:  
\begin{enumerate}[]
    \item $r_1: z>t>x>y$ (5 votes) 
    \item $r_2: y>t>x>z$ (2 votes)
    \item $r_3: x>y>t>z$ (2 votes) 
    \item $r_4: t>x>y>z$ (2 votes).  
\end{enumerate}
\par A direct computation shows that $z<_{1/2}x,y,t$. Moreover, we have 
$d^3_{KT}(r_1, V)=48$ and 
$r_1$ is the unique $3$-wise Kemeny median of the election. 
It follows that while $z$ is a Condorcet loser, it is the unique winner of the election with respect to the 3-wise Kemeny rule. Consequently, we conclude that the $3$-wise Kemeny scheme does not satisfy the Condorcet loser criterion. 
 \end{proof}

The concrete example constructed in Lemma~\ref{l:condorcet-loser-3-wise} also proves that the $3$-wise Kemeny voting scheme does not satisfy the \emph{reversal symmetry} property which  requires that if a particular candidate is a unique winner in every median, then in the mirrored election where the  preferences in each vote are inverted, the candidate cannot be the winner in any median. 

\begin{lemma}
\label{l:reversal-symmetry-3-wise}
The $3$-wise Kemeny voting scheme does not satisfy the reversal symmetry. 
\end{lemma}

\begin{proof}
    Consider again the set $C= \{x,y,z, t\}$ of 4 candidates and the following voting profile $V$ with $z$ as the unique winner with respect to the $3$-wise Kemeny voting scheme as in Lemma~\ref{l:condorcet-loser-3-wise}. 
The voting profile $V'$ of the mirrored election is then:  
\par 
\begin{enumerate}[]
    \item $r'_1: y>x>t>z$ (5 votes) 
    \item $r'_2: z>x>t>y$ (2 votes)
    \item $r'_3: z>t>y>x$ (2 votes) 
    \item $r'_4: z>y>x>t$ (2 votes).  
\end{enumerate}
\par 
Since in $V'$ the candidate $z$ wins $2+2+2=6$ votes out of $11$ votes, $z$ is the winner in more than $50\%$ of the votes. Consequently, Lemma~\ref{l:majority-criterion-k-wise} implies that $z$ must be the winner in every $3$-wise median. Hence, the reversal symmetry property is not verified by the voting profile $V$ with respect to the $3$-wise Kemeny voting scheme. The proof is complete. 
\end{proof}

\subsection{$3$-wise Smith criterion and $3$-wise Extended Condorcet criterion}

\begin{definition}
    \label{d:smith-set}
Given an election with the set of candidates $C$, its \emph{Smith set} is defined as  the smallest non-empty subset $S\subset C$ such that every candidate in $S$ is majority-preferred over every candidate in $C\setminus S$. It is clear that such a Smith set is well-defined. 
A voting scheme is said to satisfy the \emph{Smith criterion}  the winner in every consensus of an election belongs to the Smith set of that election \cite{smith}.
\end{definition}
\par 
Observe that by definition, the Smith criterion implies the Condorcet criterion. 
More generally, in \cite{truchon-XCC}, Truchon studied the so-called \emph{Extended Condorcet criterion} which says that if  there is a partition of the set of candidates $C=I \cup J$ such that for any $x$ in $C$ and any $y$ in $J$ the majority prefers $x$ to $y$ in the election, then in every median Kemeny ranking, $x$ must be ranked above $y$.
\par 

For the $3$-wise Kemeny voting scheme, we obtain in this section  the following similar space reduction result which extends notably Theorem~\ref{t:3-wise-condorcet}. 

\begin{theorem}
\label{t:XCC-3-wise-main}
Let $C$ be the set of candidates in an election. Suppose that $C=I \cup J$ is a partition of $C$ such that
\begin{enumerate}[\rm (i)]
    \item 
for all $x\in I$ and $y \in J$, we have $x \geq_{3/4} y$, 
\item $0< |I| \leq \frac{|J|+4}{2}$. 
\end{enumerate} 
\par 
Then in every $3$-wise median, we have $x>y$ for all $x \in I$ and $y \in J$.  
\end{theorem}

\par 
To establish Theorem~\ref{t:XCC-3-wise-main}, 
we shall first prove the following consequence of Theorem~\ref{t:XCC-3-wise-main} which is an 
extension of the Smith criterion for the $3$-wise Kemeny voting scheme with respect to the $3/4$-majority rule. 

\begin{theorem}
\label{t:smith-criterion-main}
Let $C$ be the set of candidates in an election. Suppose that $C=I \cup J$ is a partition of $C$ such that
\begin{enumerate}[\rm (i)]
    \item 
for all $x\in I$ and $y \in J$, we have $x \geq_{3/4} y$, 
\item $0< |I| \leq \frac{|J|+4}{2}$. 
\end{enumerate} 
Then the winner in every median of the election with respect to the $d^3_{KT}$ distance must be a candidate in $I$. 
\end{theorem}

\begin{proof} 
Let $V$ be the voting profile of the election and let $m=|V|$. Note that $I\neq \varnothing$ since $|I|>0$ by (ii). Therefore, we can suppose on the contrary that there exists a median $\pi \colon y >A>x>B$  with respect to the distance $d^3_{KT}$ such that 
    \begin{enumerate}[\rm (a)] 
        \item $A \cup\{ y\} \subset J$,
        \item $x \in I$. 
    \end{enumerate}
    \par 
Consider the modified ranking $\pi^* \colon x>y>A>B$ of the ranking $\pi$. We will show that 
\[
\Delta = d^3_{KT}(\pi^*,V) - d^3_{KT}(\pi, V) \leq 0.
\] 
\par 
Indeed, when restricted to subsets of $2$ elements, the only pairs that can contribute to $\Delta$ are $(x,y)$ and $(x,A)$. For the pair $(x,y)$, note that $x\geq_{3/4} y$ and thus $y<_{1/4}x$. Hence, the contribution of $(x,y)$ is at most $\frac{1}{4} m - \frac{3}{4} m = -\frac{1}{2}m$. Similarly, since $A \subset J$, the contribution of each pair $(x,z)$ where $z \in A$ is at most $-\frac{1}{2}m$. 
Consequently, the total contribution of subsets of size $2$ to $\Delta$ is at most
\begin{equation}
    \label{e:smith-3-wise-1} 
    \left( - \frac{1}{2} - \frac{1}{2} |A| 
    \right)m.
\end{equation}
\par 
For subsets of size $3$, observe that only  subsets of the following forms can contribute to $\Delta$: 
\begin{enumerate} [(1)] 
    \item $(x,y,A)$, 
    \item $(x,y,B)$,
    \item $(x,A,A)$,
    \item $(x,A,B)$. 
\end{enumerate}
We shall consider each case separately. 
\par 
\textbf{Case 1:} for every $z \in A$, we have $x \geq_{3/4} z$, $x \geq_{3/4}y$ thus $x>y,z$ in at least $m/2$ votes. Moreover, since  $y >x$ in at most $m/4$ votes, we deduce that the set $\{x,y,z\}$ contributes at most $(m/4 -m/2)=-m/4$ to $\Delta$. Hence, the total contribution, in this case, is bounded from the above by: 
\begin{equation}
    \label{e:smith-3-wise-case-1} 
    -\frac{1}{4}|A|m. 
\end{equation}
\par 
\textbf{Case 2:} by a similar argument, we find that the total contribution, in this case, is at most
\begin{equation}
    \label{e:smith-3-wise-case-2} 
    \left( \frac{1}{4}|B \cap I| - \frac{1}{4} |B \cap J| \right) m 
\end{equation}
since for each $z \in B \cap I$, we have $y>x$ in at most $m/4$ votes thus the term 
$ \frac{1}{4}|B \cap I|m$ and for each $z \in B \cap J$, we have $x \geq_{3/4}y$, $x \geq_{3/4}z$ thus $x>y,z$ in at least $m/2$ votes (which are of course different from the votes in which $y>x$ which are at most $m/2$ in number) whence the term 
$$ \frac{1}{4}|B \cap J|m - \frac{1}{2} |B \cap J|m= - \frac{1}{4} |B \cap J|m. 
$$
\par 
\textbf{Case 3:} similarly, the total contribution, in this case, is at most
\begin{equation}
    \label{e:smith-3-wise-case-3} 
    \left( \frac{1}{4}|(A,A)| - \frac{1}{2} |(A,A)| \right) m= - \frac{1}{4} |(A,A)|m.  
\end{equation}
\par 
\textbf{Case 4:} the total contribution is bounded by 
\begin{equation}
    \label{e:smith-3-wise-case-4} 
    \left(   \frac{1}{4} |(A,B \cap I)| - \frac{1}{4}|(A,B \cap J)| \right) m.  
\end{equation}
\par 
Let $a = |A|$, $b=|J \setminus (A \cup \{y\})|= |B \cap J|$, and $c=|I \setminus \{x\}|= |B \cap I|$ then 
\begin{align*}
   & |(A,A)|  = \frac{a(a-1)}{2}, \quad |(A, B\cap J)|= ab, \quad  |(A, B\cap I)|=ac,\\
&|I| = c+1, \quad |J|= |J \setminus (A \cup \{y\})| + |A|+ 1 = a+b+1. 
\end{align*}
\par 
To summarize, we obtain from the bounds \eqref{e:smith-3-wise-1}, \eqref{e:smith-3-wise-case-1}, \eqref{e:smith-3-wise-case-2}, \eqref{e:smith-3-wise-case-3}, and  \eqref{e:smith-3-wise-case-4} the following estimation:  
\begin{align}
\label{e:smith-3-wise-main}
    \frac{\Delta}{m} & \leq -\frac{1}{2} - \frac{1}{2}a -  \frac{1}{4}a + \frac{1}{4}c  - \frac{1}{4}b - \frac{1}{4}\frac{a(a-1)}{2} - \frac{1}{4} ab + \frac{1}{4}ac. 
\end{align}
\par 
Since $a,b,c \in  \N$ and $c+1 = |I| \leq \frac{|J|+4}{2} = \frac{a+b+5}{2}$ thus $c \leq \frac{a+b+3}{2}$ by hypothesis, we deduce from \eqref{e:smith-3-wise-main} that: 
\begin{align*}
     \frac{\Delta}{m} & \leq  -\frac{1}{2} - \frac{1}{2}a -  \frac{1}{4}a + \frac{1}{4}c  - \frac{1}{4}b - \frac{1}{4}\frac{a(a-1)}{2} - \frac{1}{4} ab + \frac{1}{4}ac\\
     & =   -\frac{1}{2} - \frac{5}{8}a  -\frac{1}{8}a^2  - \frac{1}{4}b  - \frac{1}{4} ab +  \frac{1}{4}c + \frac{1}{4}ac\\
     & \leq  -\frac{1}{2} - \frac{5}{8}a  -\frac{1}{8}a^2  - \frac{1}{4}b  - \frac{1}{4} ab +  \frac{1}{4}\frac{a+b+3}{2} + \frac{1}{4}a \frac{a+b+3}{2}\\
     & =  -\frac{1}{8}  -\frac{1}{8}a  -\frac{1}{8}b - \frac{1}{8}ab
     \\
     & < 0. 
\end{align*}
\par 
Therefore, $ d^3_{KT}(\pi^*,V) - d^3_{KT}(\pi, V)= \Delta <0$ and we obtain a contradiction since $\pi$ is a median by hypothesis. The proof is thus complete. 
\end{proof}

We are now in position to prove Theorem~\ref{t:XCC-3-wise-main}. 

\begin{proof}[Proof of Theorem~\ref{t:XCC-3-wise-main}]
 We suppose on the contrary that there exists a median $\pi$ with respect to the distance $d^3_{KT}$ which does not satisfy the conclusion of the theorem. As in the proof of Theorem~\ref{t:smith-criterion-main}, note that $I\neq \varnothing$ by (ii). Hence, we can write $\pi \colon K > y > A > x >B$ where  
\begin{enumerate}[\rm (a)] 
\item $K \subset I$  
    \item $A \cup\{ y \} \subset J$, 
    \item $x \in I$. 
\end{enumerate}
\par 
In other words, we choose $y  \in J$ to be the candidate with the highest rank in $\pi$ and $x \in I$ is the highest-ranked candidate which is ranked after $y$ in $\pi$.  
\par 
By Theorem~\ref{t:smith-criterion-main}, the candidate $y$ cannot be the winner of the median $\pi$. It follows that $K \neq \varnothing$. Let $\pi^* \colon K >x>y>A>B$. 
\par 
Let us consider the induced election $V'$ where we eliminate all the candidates in $K$ from the list of candidates and from all the votes while keeping the relative rankings of other candidates. Let $I'=I\setminus K$ then $I'\cup J$ is a partition of the set of candidates of the induced election $V'$ such that $z \geq_{3/4} t$ for all $z\in I'$ and $t \in J$ and 
\[
0< |\{x\}| \leq |I'| = |I|- |K| \leq |I| \leq \frac{|J|+4}{2}.
\]
\par 
Consequently, we infer from the proof of  Theorem~\ref{t:smith-criterion-main} that 
the rankings $\sigma \colon y>A>x>B$ and $\sigma^* \colon x>y>A>B$ satisfy 
\begin{equation}
    \label{e:t:XCC-3-wise-main-1} 
\delta_{KT}^3(\sigma^*, V') - \delta_{KT}^3(\sigma, V') <0.  
\end{equation}
\par 
We deduce from the definitions of $\sigma$, $\sigma^*$ and the relation \eqref{e:t:XCC-3-wise-main-1}   that 
\begin{align*}
    \delta_{KT}^3(\pi^*, V) - \delta_{KT}^3(\sigma, V) & = \delta_{KT}^3(\sigma^*, V') - \delta_{KT}^3(\sigma, V')  <0
\end{align*}
\par 
It follows that $\pi^*$ cannot be a median of the original election $V$ and we obtain a contradiction. The proof is thus complete. 
\end{proof}

\subsection{Smith-independence of irrelevant alternatives} 
 The well-known independence of irrelevant alternatives property (IIA) 
 requires that the relative ranking between $x$ and $y$ in every median should depend only on the relative rankings between $x$ and $y$ in every vote. The Arrow impossibility theorem (cf. \cite{arrow}) tells us that every voting scheme satisfying the IIA property, universality (uniqueness of the complete median ranking), and unanimity must be a dictatorship.  
 \par 
 The Smith-IIA property, or the Independence of Smith-dominated alternatives (ISDA), is a weaker and more reasonable voting criterion requiring that removing a candidate who is not a member of the Smith set (cf. Definition~\ref{d:smith-set}) will not change the winner of the election. 
\par 
The Smith-IIA property is known to hold true for the classical Kemeny rule and it clearly implies the Smith criterion (cf. Definition~\ref{d:smith-set}), Condorcet criterion, and the mutual majority criterion.  
Generalizing the notion of Smith sets, we define the $\alpha$-Smith sets as follows. 
\begin{definition}
\label{d:3/4-smith-set}
Let $\alpha \in [0,1]$. The \emph{$\alpha$-Smith set} of  an election over the set of candidates $C$
is the smallest non-empty subset $S\subset C$ such that for all $x \in S$ and $y \in C\setminus S$, we have $x \geq_{3/4} y$. \par   
A voting scheme satisfies the \emph{$\alpha$-Smith IIA criterion} if the winner in every consensus of an election belongs to the $\alpha$-Smith set even if we remove one or several candidates outside of the $\alpha$-Smith set. 
\end{definition}

\par 
We have the following simple observations. 
\begin{lemma}
\label{l:unique-alpha-smith-set}
Let $\alpha \in[0,1]$. Then the $\alpha$-Smith set $S_{\alpha}$ of an election $V$ over a set of candidates $C$ is unique. Moreover, $S_{\alpha}$ is the intersection of all subsets $S \subset C$ such that for all $x \in S$ and $y \in C\setminus S$, we have $x\geq_\alpha y$. 
\end{lemma}

\begin{proof}
Let $S, S'\subset C$ be such that  for all $x \in S$ and $y \in C\setminus S$,  $x\geq_\alpha y$ and such that for all $x \in S'$ and $y \in C\setminus S'$,  $x\geq_\alpha y$. Then it suffices to note that $x\geq_\alpha y$ for all 
$x \in S \cap S'$ and $y \in C\setminus (S\cap S')$.   
\end{proof}

\begin{lemma}
\label{l:alpha- beta-smith}
 Let $0 \leq \alpha\leq \beta \leq 1$.   If a voting scheme satisfies the $\alpha$-Smith IIA criterion then it also satisfies the $\beta$-Smith IIA criterion. 
\end{lemma}

\begin{proof} 
    Let $S_\alpha$ and $S_\beta$ be the $\alpha$-Smith set and the $\beta$-Smith set of the election. Then by definition, we clearly have  $S_\alpha \subset S_\beta$ (cf. Lemma~\ref{l:unique-alpha-smith-set}). Suppose 
   that a voting scheme satisfies the $\alpha$-Smith IIA criterion. Thus,  
    if we remove one or several candidates outside of  $S_\alpha$, the winner of the resulting election still belongs to $S_{\alpha}$. Since $S_\alpha \subset S_\beta$, the winner stays in  $S_{\alpha}$ and thus to $S_\beta$ if we remove one or several candidates outside of  $S_\beta$. Hence, the election also satisfies the $\beta$-Smith IIA criterion. 
\end{proof}
\par 
While the $3$-wise Kemeny voting scheme satisfies the $3/4$-Smith criterion (Theorem~\ref{t:smith-criterion-main}), 
the following example shows that the $3$-wise Kemeny voting scheme does not satisfy the $3/4$-Smith-IIA property and thus fails the Smith-IIA property in particular (cf. Lemma~\ref{l:alpha- beta-smith} as $3/4>1/2$).   

\begin{example} 
\label{ex:3/4-smith-iia}
 Let $A$, $B$ denote the blocks 
 $x_3>x_4>x_5$ and $x_6>x_7>x_8$  respectively. 
Let us consider the following voting profile $V$ consisting of $4$ votes over $8$ candidates $x_1, \dots, x_8$:  
 \begin{enumerate}[]
    \item $r_1: B>x_2>x_1> A$ \quad ($1$ vote) 
    \item $r_2: x_1>A>x_2>B$ \quad ($1$ vote) 
    \item 
   $r_3: A> x_2>x_1>B$ \quad ($1$ vote)
    \item $r_4: x_2>x_1>A>B$ \quad ($1$ vote) 
\end{enumerate}
\par 
An exhaustive computation shows that the only $3$-wise median of the election $V$ is $\pi^*\colon x_1>A>x_2>B$ whose $3$-wise distance to $V$ is 114. In particular, $x_1$ is the unique winner of the election. 
\par 
Note that 
$x_2 \geq_{3/4}x_1$ and $x_1\geq_{3/4}x_i$ for all $i=3,4,\dots,8$. Moreover, for every $i =3,4,5$, the candidate $x_2$ is ranked before $x_i$ in exactly half of the votes. 
Consequently, it is not hard to see that the $3/4$-Smith set of the election $V$ is $S=\{x_1,x_2,x_3,x_4,x_5\}$. 
\par 
However, if we remove the candidates $x_6,x_7,x_8$, we obtain the following election: 
 \begin{enumerate}[]
    \item $r_1': x_2>x_1> A$ \quad ($1$ vote) 
    \item $r_2': x_1>A>x_2 $ \quad ($1$ vote) 
    \item 
   $r_3': A> x_2>x_1$ \quad ($1$ vote)
    \item $r_4': x_2>x_1>A$ \quad ($1$ vote) 
\end{enumerate}
whose unique $3$-wise Kemeny median is $\sigma^* \colon x_2>x_1>x_3>x_4>x_5$ and thus we obtain a new unique $3$-wise winner $x_2 \neq x_1$.   We conclude that the $3/4$-Smith-IIA property fails for the election $V$ under the $3$-wise Kemeny voting scheme. 
\end{example}

\section{Unanimity and the unique winner property} 

The original Unanimity Theorem \cite[Proposition~5]{setwise} guarantees the relative ordering of a pair of candidates in every final $k$-wise ranking if all votes agree on the same preference of that pair. However, given $k \geq 2$, the Unanimity Theorem alone does not allow us to arrive at the same conclusion whenever there exist two votes which have different preferences on a specific pair of candidates $(x,y)$, even if $x >y$ in virtually all votes.  
\par 
In the case $k=3$, we shall establish  Theorem~\ref{t:3-wise-unanimity-general} in the next section to solve the above issue quantitatively and prove the possibility to manipulate the ordering of a pair of candidates in every election using the 3-wise Kemeny voting scheme together with a simple strategy to achieve the desired manipulation. More specifically, suppose that in an election using the 3-wise Kemeny rule, we want to manipulate the relative  preference of two candidates $x,y$ so that $x>y$ in every final ranking, it suffices to make sure that $x>y$ in $g(n)\times 100\%$ of the votes. Here, $g \colon \N \to ]0,1[$  is an increasing function defined in Theorem~\ref{t:3-wise-unanimity-general}.

\subsection{$3$-wise Extended Pareto efficiency (Unanimity)} 

In the vein of the Extended Always theorem (Theorem~\ref{t:extended-always}), 
we have the following generalization of the unanimity property \cite[Proposition 5]{setwise} also known as the Pareto efficiency for the $k$-wise Kemeny rule when $k=3$:  

\begin{theorem} 
\label{t:3-wise-unanimity-general}
Let $x,y$ be candidates in an election with $n \geq 2$ candidates. Suppose that $x \geq_\alpha y$ for some $\alpha \in [0,1] $ such that 
\begin{equation*}
 \alpha > g(n)= 1 -  \frac{1}{n^2-3n+4}. 
\end{equation*}
 \par 
 Then $x>y$ in every median with respect to the $3$-wise Kemeny rule. 
\end{theorem}

\par 
To illustrate, the following table gives us several values of $g(n)$: 
\begin{center}
\begin{tabular}{ | M{2.5em} |  M{1em} | M{1cm}  | M{1cm}  | M{1cm} | M{1cm} | M{1cm}| M{1cm}| M{1cm}| } 
  \hline
  $n$ & 2 & 3 & 4 & 5 & 6 & 8 & 10 & 12 \\ 
  \hline
  $g(n)$ & 0.5 & 0.75 & 0.875 & 0.929  &   0.955 &   0.977 &   0.987 & 0.991 \\ 
  \hline
\end{tabular}
\end{center}

\begin{proof}
Let $C$ be the set of all candidates then $\vert C \vert = n$. 
Suppose on the contrary that there exists a Kemeny ranking $r$ of the election with respect to the distance $d^3_{KT}$ in which $x <y$. Let $L, K, R$ be respectively the ordered sets of all candidates ranked before $y$, between $x$ and $y$, and after $x$ in the ranking $r$. In other words, the ranking $r$ can be written as  $L > y > K >x >R$. 
\par 
Let $V$ be the multiset of all votes and $v= \vert V \vert$. 
We consider the ranking $r^*$ obtained from $r$ by simply exchanging the positions of the candidates $x$ and $y$ while keeping the positions  of all other candidates. Hence, 
the ranking $r^*$ is $L > x > K >y >R$. 
In the sequel, we are going to show that $\Delta = d^3_{KT}(r^*, V)- d^3_{KT}(r,V) <0$ to obtain a contradiction. 
\par 
By the inequality \eqref{e:t:extended-always} in the proof of Theorem~\ref{t:extended-always}, the contribution to the score difference $\Delta$ between $r^*$ and $r$, when restricted to $\Delta^2(C)$, namely,  subsets of $C$ of cardinality at most $2$,  is bounded by: 
\begin{equation}
\label{e:extended-3-wise-always-1}
    (1-2\alpha + (n-2)(1-\alpha))v. 
\end{equation}
\par 
We next consider subsets $S$ of $C$ of cardinality $3$. Since $r$ and $r^*$ are identical when restricted to $C \setminus \{x,y\}$, the restriction to $S$ can contribute to $\Delta$ only if $\mathrm{top}_S(r) \neq \mathrm{top}_S(r)$, which can happen only if $S$ contains $x$ or $y$, i.e., $S \cap \{x,y\} \neq \varnothing$. We thus distinguish three possibilities. 
\par 
 \underline{\textbf{Case 1:}} $x, y \in S$. Then we can write $S= \{x,y,z\}$ for some $z \in C \setminus \{x,y\} = L \cup K \cup R$. If $z \in L$ then $\mathrm{top}_S(r) = \mathrm{top}_S(r^*)= z$ and thus $S$ contributes nothing to $\Delta$. 
If $z \in K \cup R$ then $\mathrm{top}_S(r) = y$ and $\mathrm{top}_S(r^*) = x$. In this case, $S$ also contributes at most $0$ to $\Delta$ when computed with at least $\alpha v$ votes where $x >y$. We conclude that the total contributions to $\Delta$ of the subsets $S$ of the form $\{x,y,z\}$, where $z \in C \setminus \{x,y\}$, is at most 
\begin{equation}
    \label{e:extended-3-wise-always-case-1} 
   ( \vert K \vert + \vert R \vert) (v- \alpha v) \leq (n-2)(1- \alpha)v.  
\end{equation}
\par 
\underline{\textbf{Case 2:}} $x \in S$ and $y \notin S$. Then we can write $S=\{x,z,t\}$. Let $\pi \in V$ be a vote such that $x>y$ in $\pi$ and such that the contribution of $S$ to $\Delta$ when computed with $\pi$ is one. We can suppose without loss of generality that $z>t$ in $\pi$. Then it follows from the definition of $r^*$ that $z,t \notin L$ and $\mathrm{top}_S(\pi) = \mathrm{top}_S(r)=z$
. Consider the set $T(S) = \{y,z,t\}$. Then the  contribution of $T(S)$ to $\Delta$ when computed with $\pi$ is $-1$ since 
$\mathrm{top}_{T(S)}(\pi) = \mathrm{top}_{T(S)}(r^*)=z$ while $\mathrm{top}_{T(S)}(r) = y \neq z$.  Note that $T(S)$ and $S$ are distinct and  uniquely determine each other and the combined contribution of $S$ and $T(S)$ to $\Delta$ when computed with $\pi$ is zero.
\par 
Since there exist at most $(1- \alpha)v$ votes $\pi$ where we do not have $x>y$ and since there are at most $\frac{(n-2)(n-3)}{2}$ subsets $S$ in Case 2, we conclude that the total contribution to $\Delta$ of the subsets $S$ in this case is at most 
\begin{equation}
    \label{e:extended-3-wise-always-case-2} 
    \frac{(n-2)(n-3)}{2} (1-\alpha)v. 
\end{equation}
\par 
\underline{\textbf{Case 3:}} $y \in S$ and $x \notin S$. Then  $S=\{y,z,t\}$ for some $z,t \in C \setminus \{x,y\}$. Fix a vote $\pi \in V$ in which $x>y$ and such that the contribution of $S$ to $\Delta$ when computed with $\pi$ is one. We can suppose  that $z>t$ in $r$ and $r^*$. Observe that $z,t \notin L$, $\mathrm{top}_S(\pi)=\mathrm{top}_S(r)=y$ and 
$\mathrm{top}_S(r^*)=x$. As in Case 2, let $T(S)=\{x,z,t\}$. As $x>y$ in $\pi$, the contribution of $T(S)$ to $\Delta$ when computed with $\pi$ is $-1$ since we have 
$\mathrm{top}_{T(S)}(\pi) =  \mathrm{top}_{T(S)}(r^*)=x$  while $\mathrm{top}_{T(S)}(r) = z \neq x$. 
Therefore, we deduce from the same argument as in Case 2 that the total contribution to $\Delta$ of the subsets $S$ in Case 3 is at most 
\begin{equation}
    \label{e:extended-3-wise-always-case-3} 
    \frac{(n-2)(n-3)}{2} (1-\alpha)v. 
\end{equation}
\par 
To summarize, we conclude from \eqref{e:extended-3-wise-always-1},  \eqref{e:extended-3-wise-always-case-1}, \eqref{e:extended-3-wise-always-case-2}, and  \eqref{e:extended-3-wise-always-case-3} the following estimation on the 3-wise Kemeny  score difference:  
\begin{align*}
    \frac{\Delta}{v} & \leq 1- 2 \alpha + 2(n-2)(1-\alpha) + (n-2)(n-3) (1-\alpha) \\
    & = n^2-3n+3 - \alpha (n^2-3n+4)
\end{align*}
\par 
As $v >0$, we deduce the following implication:
\begin{align*}
  \alpha > \frac{n^2-3n+3}{n^2-3n+4} = g(n)  \Longrightarrow     \Delta <0. 
\end{align*}
\par 
Therefore, 
whenever $\alpha > g(n)$, we have $\Delta <0$ which implies that $r$ cannot be a Kemeny ranking with respect to the 3-wise Kemeny rule. Hence, we must have $x>y$ in every 3-wise Kemeny ranking.  
\end{proof}

\subsection{Optimality}

In the following question, we try to find the best sufficient condition under which the $3$-wise Kemeny voting scheme guarantees the uniqueness of the winner of an election.    
\begin{question}
What is the smallest number $\alpha \in [0,1]$  such that if $x \geq_\alpha y$ for all candidates $y \neq x$ in an election then $x$ is the unique  winner in every median with respect to the 3-wise Kemeny rule? 
\end{question} 
\par 
By Theorem~\ref{t:3-wise-condorcet}, we know that  $\alpha$ exists and must be smaller than or equal to $3/4$. Moreover, we know by \cite[Proposition 3]{setwise} that $\alpha > 1/2$. Hence, $\alpha \in ]1/2, 3/4]$.  
With only three candidates, we have  the following lemma which shows that $\alpha >3/5$.

\begin{lemma}
\label{c:lower-bound-3/5-theorem-5.1}
Let $\alpha \in [0,1]$ be the smallest number such that if $x \geq_\alpha y$ for all candidates $y \neq x$ in an election with three candidates then $x$ is the winner in every median in the 3-wise Kemeny voting scheme. Then $\alpha \in ]3/5, 3/4]$. 
\end{lemma}

\begin{proof} 
Let $C= \{x,y,z\}$ be a set of 3 candidates and consider the following election with voting profile:  
\begin{enumerate}[]
    \item $r_1:  z>x>y$ (4 votes) 
    \item $r_2: y>x>z$ (4 votes)
    \item $r_3: x>y>z$ (1 vote) 
    \item $r_4: x>z>y$ (1 vote).  
\end{enumerate}
\par Then we can check by a direct computation that $x\geq_{3/5}y$ and $x\geq_{3/5}z$. Moreover,  we have:  
\[
d^3_{KT}(r_1,V) = d^3_{KT}(r_{2},V)=d^3_{KT}(r_3,V) = d^3_{KT}(r_4,V) = 21, 
\]
and $r_1, r_2,r_3,r_4$ are all the $3$-wise medians of the election. 
It follows that $x$ is not 
 the unique winner of the election with respect to the 3-wise Kemeny voting scheme. 
\end{proof}

For elections with four  candidates, we can prove the following improved lower bound $5/8$.

\begin{lemma}
\label{c:lower-bound-5/8-theorem-5.1}
Let $\alpha \in [0,1]$ be the smallest number such that if $x \geq_\alpha y$ for all candidates $y \neq x$ in an election with four candidates then $x$ is the winner in every median in the 3-wise Kemeny voting scheme. Then $\alpha \in ]5/8, 3/4]$. 
\end{lemma}

\begin{proof}
    Let $C=\{x,y,z,t\}$ be a set of $4$ candidates.  Consider the following voting profile $V$: 
    \begin{enumerate}[]
    \item $r_1: z>t>x>y$ (3 votes) 
    \item $r_2: y>x>t>z$ (3 votes)
    \item $r_3: t>z>x>y$ (2 votes) 
\end{enumerate}
\par 
We can verify by a direct computation that $r_1$ and $r_3$ are the only $3$-wise medians of the election $V$ and 
\[
d^3_{KT}(r_1, V) = d^3_{KT}(r_3, V) = 36. 
\]
\par
In particular, $t$ is not the unique winner of the election $V$ while we have $t>_{5/8}x,y,z$ in $V$. The proof is complete. 
\end{proof}

\begin{example}
Let $V$ be a voting profile of an election with $9$ candidates $\{1,2, \dots, 9\}$ as follows: 
\begin{enumerate}[]
    \item $r_1 \colon x_3>x_4>x_1>x_2>x_5>x_6>x_7 >x_8>x_9 $ (15 votes) 
    \item $r_2 \colon x_9>x_8>x_7>x_6>x_5 >x_2>x_1>x_4>x_3$ (15 votes)
    \item $r_3 \colon  x_4>x_3>x_1>x_2 >x_5>x_6>x_7 >x_8>x_9 $ (13 votes)
\end{enumerate}
 \par 
 A direct computation shows that $
 \pi^*=r_1$ is the unique $3$-wise median of the election and $d^3_{KT}(\pi^*,V)= 1904$. Therefore, the candidate 
$x_4$ wins the head-to-head competition over every other candidate by the ratio $28/43=0.6511$ but loses the election to the candidate $x_3$.  
\end{example}

In general, we establish the following estimation which implies a surprising phenomenon: in contrast to the Condorcet criterion for the classical $2$-wise Kemeny scheme, even the $2/3$ majority in every duel is not enough to guarantee that a candidate will win the election according to the $3$-wise Kemeny voting scheme. Consequently, it is harder to win an election with respect to the $3$-wise Kemeny voting scheme than the $2$-wise Kemeny voting scheme.

\begin{theorem}
\label{c:lower-bound-2/3-theorem-5.1}
Let $\alpha \in [0,1]$ be the smallest number such that if $x \geq_\alpha y$ for all candidates $y \neq x$ in an election then $x$ is the winner in every median in the 3-wise Kemeny voting scheme. Then $\alpha \in [2/3, 3/4]$. 
\end{theorem}

\begin{proof}
Let $0< \beta <2/3$ be an arbitrary real number and let us fix an integer $n\geq 33$  such that  
\begin{align}
\label{e:c:lower-bound-2/3-theorem-5.1-0}
n > \frac{8- 8 \beta}{2-3\beta} \quad \text{or equivalently} \quad \frac{2n-8}{3n-8} >\beta. 
\end{align}
 \par 
Let $C=\{x,y,z,t,p_1,p_2,p_3,u_1, \dots, u_n\}$ be a set of $n+7$ candidates and consider the following voting profile $V$: 
 \begin{enumerate}[]
    \item $r_1: z>t>x>y>u_1>u_2>\dots>u_n>p_1>p_2>p_3$ \quad ($n$ votes) 
    \item $r_2: u_1>u_{2}> \dots > u_n > y>x>t>z>p_1>p_2>p_3$ \quad ($n$ votes)
    \item $r_3: t>z>x>y>u_1>u_2>\dots>u_n>p_1>p_2>p_3$ \quad ($n-8$ votes) 
\end{enumerate}
\par 
We claim that $r_1$ is the unique $3$-wise median of the election $V$. Indeed, let $\pi$ be a $3$-wise median of the election then by Theorem~\ref{t:3-wise-unanimity-general}, we have 
$u_1>u_2>\dots >u_n>p_1>p_2>p_3$ in  $\pi$ and $p_1, p_2, p_3$ occupy the last three positions in $\pi$. 
\par 
Let us show that in $\pi$, we have  $x,y,z,t >u_i$ for every $i=1,\dots,n$. Suppose on the contrary that for some $i \in \{1,\dots,n\}$, the candidate $u_i$ is ranked immediately before a candidate $s \in \{x,y,z,t\}$ in $\pi$. We choose $u_i$ such that it has the lowest rank among such $u_i$.  Hence, $\pi$ is of the form 
\[
\pi \colon A> u_i >s>K >B>p_1>p_2>p_3
\]
where $K \subset \{s,y,z,t\}\setminus\{s\}$ and $B \subset \{u_1,\dots, u_n\}\setminus\{u_i\}$. 
\par 
By exchanging the positions of $u_i$ and $s$ in $\pi$, we obtain the following ranking: 
\[
\pi^* \colon A> s> u_i >K >B>p_1>p_2>p_3
\]
\par 
We will show that $\Delta = d^3_{KT}(\pi^*, V) - d^3_{KT}(\pi, V)$. The only subsets of at most $3$ elements that can contribute to $\Delta$ are of the form $\{u_i,s\}$, $\{u_i, s, u\}$ where $u \in  K \cup B \cup\{p_1, p_2, p_3\}$. By comparing with the rankings $r_1$, $r_2$, $r_3$ respectively, we deduce that the contribution of $(u_i,s)$ to $\Delta$ is 
\begin{equation}
    \label{e:c:lower-bound-2/3-theorem-5.1-1}
-n+n-(n-8)=8-n. 
\end{equation} 
\par 
Let $u \in B \cup\{p_1, p_2, p_3\}$. Then by comparing $\pi^*$ and $\pi$ with $r_1$, $r_2$, $r_3$ respectively, the contribution of $(u_i,s, u)$ to $\Delta$ is at most 
\begin{equation}
    \label{e:c:lower-bound-2/3-theorem-5.1-2}
-n+n-(n-8)=8-n. 
\end{equation}
\par 
Similarly, we find that the contribution of each subset $\{u_i,x, u\}$, where $u \in K$, to $\Delta$ is at most 
\begin{equation}
     \label{e:c:lower-bound-2/3-theorem-5.1-3}
0+n + 0=n. 
\end{equation}
\par 
Combining the bounds \eqref{e:c:lower-bound-2/3-theorem-5.1-1}, \eqref{e:c:lower-bound-2/3-theorem-5.1-2}, \eqref{e:c:lower-bound-2/3-theorem-5.1-3} and notices that $n \geq 33$, $|K|\leq 3$,  we deduce   the following estimation:  
\begin{align*}
    \Delta  & \leq 8-n+ (8-n)(|B|+3) + n|K| \\
    & \leq 8-n + 3(8-n)+ 3n\\
    & \leq  32-n  <0. 
\end{align*}
\par 
Hence, $d^3_{KT}(\pi^*, V) - d^3_{KT}(\pi, V)=\Delta <0$, which contradicts the hypothesis that $\pi$ is a $3$-wise median. Therefore, we must have   $x,y,z,t >u_i$ in  $\pi$. Consequently, every median $\pi$ is of the form 
\[
\pi \colon D > u_1 >\dots >u_n >p_1>p_2>p_3. 
\]
where $D$ is a permutation of $\{x,y,z,t\}$. We claim that $z$ is the winner of $\pi$. Otherwise, let $s \in \{x,y,t\}$ be the candidate which precedes $z$ in $\pi$. Let $\pi^{**}$ be the ranking obtained from $\pi$ where we exchange the positions of $s$ and $z$. Let $\delta = d^3_{KT}(\pi^{**}, V) - d^3_{KT}(\pi, V)$ then the only subsets that can contribute to $\delta$ are 
$\{z,s\}$, $\{z,s,u\}$ where $u \in C \setminus \{s,z\}$. 
By comparing with $r_1$, $r_2$, and $r_3$ respectively,  the contribution of $\{s,x\}$ to $\delta$ is at most: 
\begin{equation}
\label{e:c:lower-bound-2/3-theorem-5.1-4}
    -n + n +n-8 = n-8. 
\end{equation}
\par 
Similarly, the total contribution to $\delta$ of all the subsets of the form   $\{z,s,u\}$ with $u \in C \setminus \{s,z\}$ is at most: 
\begin{equation}
\label{e:c:lower-bound-2/3-theorem-5.1-5}
    -(n+3)n + 5n + (n+5)(n-8) = -n - 40.  
\end{equation}
\par 
In \eqref{e:c:lower-bound-2/3-theorem-5.1-5}, the term $-(n+3)n$ results from the the comparison with $n$ votes of the same  ranking $r_1$ and the  $(n+3)$ subsets $\{z,s,u\}$ where $u \in C \setminus \{x,y,z,t\}$. The second term $5n$ results from the $n$ votes $r_2$ and at most 5 subsets of the form $\{x,z,u\}$ where $u \in \{x,y,z,t,p_1,p_2,p_3\}\setminus \{z,s\}$. Likewise, the third term  $(n+5)(n-8)$ results from $(n-8)$ votes of the same ranking $r_3$ and $(n+5)$ subsets $\{z,s,u\}$ with  $u \in \{u_1, \dots, u_n\}$. 
\par 
In summary, we can bound $\delta$ using  \eqref{e:c:lower-bound-2/3-theorem-5.1-4} and \eqref{e:c:lower-bound-2/3-theorem-5.1-5} as follows: 
\begin{align*}
d^3_{KT}(\pi^{**}, V) - d^3_{KT}(\pi, V) = \delta \leq n-8 - n- 40 = -48<0,
\end{align*}
which is a contradiction since $\pi$ is a $3$-wise median. We conclude that $z$ must be the unique winner of the election with respect to the $3$-wise Kemeny voting scheme. 
\par 
On the other hand, a direct computation shows that   
$t>u$, for all candidate $u \in C \setminus \{t\}$, in at least 
 $2n-8$ out of $3n-8$ votes, thus in a fraction at least 
 \[
 \frac{2n-8}{3n-8} > \beta 
 \]
 of the votes (see \eqref{e:c:lower-bound-2/3-theorem-5.1-0}). 
Therefore, we have constructed an election in which a candidate $t$ satisfies $t>_{\beta} u$ for every other candidate $u \neq t$ but $t$ is not the winner in any $3$-wise median of the election. 
 The proof is thus complete. 
\end{proof}

\begin{remark}
\label{r:win-better-than-draw}
  Lemma \ref{c:lower-bound-3/5-theorem-5.1} seems to suggest that for a candidate, it is particularly more meaningful to win several votes and lose all other votes (e.g. rank last) than to be in the middle of almost all votes. A distant but maybe related situation is the pointing system in football where three draws (3 points) equal to one win (3 points) and two loses (0 points).    
\end{remark}

\section{Remarks on $3$-wise $3/4$-majority rule for elections with few candidates} 

By \cite{kien-sylvie-MOT}, 
the $3/4$-majority rule holds only 
for elections of at most $5$ candidates for the $3$-wise Kemeny voting scheme.  
Nevertheless, for elections with $6$ candidates, we have the following weak form of the $3$-wise $3/4$-majority rule. 

\begin{theorem}[\cite{kien-sylvie-MOT}]
    Let $C$ be a set of $6$ candidates. Suppose that in an election over $C$, we have a partition $C= A \cup \{x\} \cup B$ where $|A| \leq 2$ such that $y \geq_{3/4} x$ for all $y \in A$ and $x \geq_{3/4}z$ for all $z \in B$. Then the election satisfies the $3$-wise $3/4$-majority rule. 
\end{theorem}

Recall that by Theorem~\ref{t:3-wise-condorcet}, the $3$-wise $3/4$-majority holds for all non-dirty candidates who win every duel by the ratio $3/4$. For non-dirty candidates who lose the head-to-head competition to exactly one other candidate, we have the following useful results.
When there is no possible confusion, we drop the $>$ symbol in a ranking for simplicity. 

\begin{lemma}[\cite{kien-sylvie-MOT}]
 Let $C=\{x,z\}\cup J$ be a partition of the set of candidates of an election $V$ such that $z\geq_{3/4}x$ and $x\geq_{3/4} y$ for all $y \in J$.  
 Then the following properties hold: 
 \begin{enumerate}[\rm (a)] 
     \item For all partitions $J=A\cup B \cup C$ where $A,B,C$ are ordered sets  such that $B\neq \varnothing$, we have: 
     \[
     d^3_{KT}(AzBxC, V) > d^3_{KT}(AzxBC, V). 
     \]
     \item For all partitions $J=  A \cup B$ where $A,B$ are ordered sets with $A \neq \varnothing$, we have: 
     \[
     d^3_{KT}(AzxB, V) > d^3_{KT}(zxAB, V).
     \]
 \item 
For all partitions $J= A \cup B \cup C$ where $A,B,C$ are ordered sets such that $|B \cup C|=5$ and $|B|\leq 2$, we have: 
 \[
 d^3_{KT}(AxBzC, V) > d^3_{KT}(AzxBC, V) 
 \]
 
 \end{enumerate}
\end{lemma}
 
\section{The $5/6$-majority rule for the $3$-wise Kemeny voting scheme} 

 We first establish the following weak $5/6$-majority rule for the $3$-wise Kemeny voting scheme which in a sense extends  the $3/4$-majority rule for the classical Kemeny voting scheme. Notably, our result is particularly useful when we know that by other space reduction techniques, every median of the election must be of some special form.

\begin{theorem}
\label{l:5/6-majority-3-wise-main}
Let 
$x$ be a non-dirty candidate in an election 
 with respect to the $5/6$-majority rule and let $I= \{z \neq x \colon z\geq_{5/6}x\}$. 
Let $r$ be a median with respect to the distance $d^3_{KT}$ such that $z>x$ in $r$ for all $z \in I$. Suppose that  
\begin{equation}
    \label{e:5/6-majority-3-wise-condition-main}
 |I|(|I|-4) \leq 3 |\{ z \colon x> z \text{ in } r\}|.  
\end{equation} 
Then for every candidate  $y\neq x$ with $x \geq_{5/6} y$, we have $x > y$  in $r$. 
\end{theorem}

We also obtain below another version of our $5/6$-majority rule for the $3$-wise Kemeny voting scheme which complements the  range of applications of  Theorem~\ref{l:5/6-majority-3-wise-main} including applications in multi-winner 
 voting systems.

\begin{theorem}
\label{l:5/6-majority-3-wise}
Let $\lambda >0$,  $s \geq \frac{5\lambda +1}{6\lambda+1}$. 
Let 
$x$ be a non-dirty candidate in an election with respect to the $s$-majority rule.  Let $r$ be a median with respect to the distance $d^3_{KT}$ such that  if $z \geq_s x$ then $z>x$ in $r$. Suppose that 
\begin{equation}
    \label{e:5/6-majority-3-wise-condition}
|\{ z \colon x> z \text{ in } r \}|\geq \lambda |\{ z \neq x \colon z\geq_s x \}|. 
\end{equation} 
Then for every candidate  $y\neq x$ with $x \geq_{s} y$, we have $x > y$  in $r$. 
\end{theorem}

The conditions \eqref{e:5/6-majority-3-wise-condition-main} and  \eqref{e:5/6-majority-3-wise-condition} are easily satisfied when $x$ is a strong candidate in the election or when for example by our extension of the Always theorem (cf. Theorem~\ref{t:3-wise-unanimity-general}) we know 
that $x$ wins against a significant number, says, $N$ candidates in every median so that the quotient  
$N / |\{ z \colon z\geq_s x \}|$ is large enough. Another simple but important method  to obtain  \eqref{e:5/6-majority-3-wise-condition-main} or   \eqref{e:5/6-majority-3-wise-condition} for large $\lambda$ is the "dilution method" which consists of  introducing a large enough number of weak candidates to the election. This strategy reflects the \emph{dependence of irrelevant alternatives} of $k$-wise Kemeny voting schemes as already observed in \cite{setwise}, i.e., the relative orders of two candidates in a median can depend on the presence of other candidates.

\begin{proof}[Proof of Theorem~\ref{l:5/6-majority-3-wise-main}] 
Let $C$ be the set of all candidates in the election. Let $V$ be the set of all votes and $m=\vert V \vert$. 
Suppose on the contrary that there exists a candidate $y\neq x$ such that  
\begin{enumerate}[\rm (C)]
    \item $x \geq_{s} y $ but $y>x$ in the ranking $r$ where $s=5/6$. 
\end{enumerate}
We can moreover suppose without loss of generality that $y$ is the lowest-ranked candidate among all candidates satisfying Condition (C). 
\par 
In the ranking $r$, let $A$, $K$, $B$ be respectively the ordered set of all candidates ranked before $x$, between $x$ and $y$, and after $y$. Therefore, we can write the ranking $r$ as $A>y > K >x >B$. Since $x$ is a non-dirty candidate with ratio $s$, we have a partition $K= L \coprod R$ where  
\begin{equation*}
    L = \{z \in K \colon z \geq_{s} x\} \subset I, \quad \quad  R = \{z \colon z <_{s} x\}
\end{equation*}
and the relative orders of candidates in $L$ and $R$ are induced by the relative orders in $K$.  
Observe also that $x \geq_{s} z$ for every candidate $z \in B$ and that by the choice of $y$, we must have $R=\varnothing$ (otherwise, we can replace $y$ by  any candidate $t \in R$ which will have a lower rank than $y$ in $r$).   
\par 
Consider the modified ranking $r^*$  in which  $A>L > x > y >B$. When restricted to $\Delta^2(C)$, the only pairs of candidates which can contribute to the score difference $\Delta = d^3_{KT}(r^*, V) - d^3_{KT}(r, V)$  are $(x,y)$, $(y, L)$, where $(y, L)$ means any pair $(y,z)$ with $z \in L$. By the proof of Theorem~\ref{l:extended-majority} (where in the inequality \eqref{e:l:extended-majority-main}, we take $q=\varepsilon =5/6$, $s=3/2$, and $a=|R|=0$), such contribution is strictly negative. 
\par 
For the contribution to the score difference $\Delta$ of subsets $S$ of exactly three candidates, we distinguish three cases with possibly non-trivial contribution: 
\begin{enumerate} [\rm (i)] 
    \item $S= (x,y, L \cup B)$; 
    \item $S = (y, L, L \cup B)$. 
\end{enumerate}
\par 
Let $z \in L$. Then as $z\geq_s x$ and $x\geq_s y$, we deduce that $z >x,y$ in at least $(2s-1)m$ votes as well as in the ranking $r^*$. 
Since $y >z ,x$ in $r$, it follows that the (unordered) subset $S=(x,y,z)$ contributes at most $m-(2s-1)m=(2-2s)m$ to $ d^3_{KT}(r^*, V)$. On the other hand, since  $x>_{s}y$ but $y>x,z$ in $r$, the set $S=(x,y,z)$  contributes at least $sm$ to 
$ d^3_{KT}(r, V)$. Consequently, the total contribution of subsets of the form $S=(x,y,L)$ to the score difference $\Delta$ is at most
\[
(2-2s - s)m \vert L \vert  = (2-3s) m \vert L \vert.
\]
\par 
Since $x\geq_s z$ for all $z \in  B$, a similar argument as above shows the subsets of the form $S=(x,y, B)$ contribute to $\Delta$ at most $(2-3s) m \vert B \vert$. Therefore, in case (i), the total contribution of  $S=(x,y,L\cup B)$ to $\Delta$ is no more than
\begin{equation}
\label{e:5/6-majority-rule-case-1}
(2-3s) m \vert L \cup B \vert = -\frac{1}{2} m \vert L \cup B \vert \leq 0. 
\end{equation}
\par 
For case (ii), let $z \in L$ and $t \in L \cup B$. Then since $x \geq_s y$ and $z \geq_s x$, we have 
$z>y$ in at least $(2s-1)m$ votes. Consequently, $y>z$ in at most $m-(2s-1)m=(2-2s)m$ votes and thus the total contribution of subsets  $S=(y,L, L \cup B)$ to $ d^3_{KT}(r^*, V)$ is at most 
\begin{equation}
\label{e:5/6-majority-rule-case-3-a}
(2-2s)m|(L, L\cup B)|.
\end{equation}

On the other hand, suppose that $t \in B$. Then $z\geq_{2s-1}  t$ since $z \geq_s x$ and $x \geq_s t$. Similarly, $z \geq_{2s-1} y$ as $z \geq_s x$, $x\geq_s y$. Consequently, we deduce that $z >y,t$ in at least $(2(2s-1)-1)m=(4s-3)m$ votes and thus the total contribution of subsets  $S=(y,L, B)$ to $ d^3_{KT}(r, V)$ is at least 
\begin{equation}
\label{e:5/6-majority-rule-case-3-b}
(4s-3)m |(L, B)|. 
\end{equation}
\par 
To summarize, we deduce from the  inequalities \eqref{e:5/6-majority-rule-case-1}, \eqref{e:5/6-majority-rule-case-3-a}, \eqref{e:5/6-majority-rule-case-3-b}, and the fact $s=5/6$ that $\Delta$ is bounded from the above as follows:  
\begin{align}
\label{l:5/6-majority-3-wise-main-delta}
\frac{\Delta}{m} & 
<  (2-3s)  \vert L \cup B \vert   + (2-2s)|(L, L\cup B)|  - (4s-3) |(L, B)| 
\\
& =  
(2-3s)  \vert L \cup B \vert  + (2-2s)|(L,L)|  +(5-6s)  |(L,  B)| \nonumber\\
& = (2-3s)  (\vert L \vert +  \vert B \vert)   + (2-2s)\frac{|L|(|L|-1)}{2}  +(5-6s)  |L|.|B| \nonumber\\ 
& = \frac{|L|(|L|-4)}{6} - \frac{|B|}{2}  \quad \quad\quad \left(\text{as }s =5/6\right)\nonumber 
\\
& \leq 0 \quad \quad\quad \quad \quad\quad \quad \quad\quad  \left(\text{by hypothesis}\right). \nonumber 
\end{align}
\par 
Hence, $\Delta <0$ and it follows that  $ d^3_{KT}(r^*, V) < d^3_{KT}(r, V)$, which is a contradiction to the hypothesis that $r$ is a median with respect to the distance $d^3_{KT}$.   The proof is thus complete. 
\end{proof}   

To illustrate, we will formulate several  $5/6$-majority rules for elections with a small number of candidates with respect to the $3$-wise Kemeny voting scheme. More specifically, we obtain the following stronger results than 
Theorem~\ref{l:5/6-majority-3-wise} for small elections. The results might also be useful in the case when Theorem~\ref{t:3-wise-unanimity-general}  and Theorem~\ref{t:3-wise-condorcet} cannot be applied (note that $g(6)=0.995>5/6$ in Theorem~\ref{t:3-wise-unanimity-general}). 

\begin{corollary}
\label{l:5/6-majority-3-wise-6-candidates}  
Let $x$ be a non-dirty candidate with respect to the $5/6$-majority rule in an election with at most $6$ candidates. Let  $r$ be a median with respect to the $3$-wise Kendall-tau distance such that if $z \geq_{5/6} x$ then $z>x$ in $r$.
Then for every $y\neq x$ with $x \geq_{5/6} y$, we have $x > y$  in $r$.   
\end{corollary}

\begin{proof}
We adopt the proof by contradiction of Theorem~\ref{l:5/6-majority-3-wise-main}. 
By substituting $s=5/6$ in the third line in the inequality \eqref{l:5/6-majority-3-wise-main-delta}, we deduce that 
\begin{align*}
    \frac{\Delta}{m} < \frac{|L|(|L|-4)}{6} - \frac{|B|}{2} \leq 0
\end{align*}
since $|B| \geq 0$ and $|L| \leq |C \setminus \{x,y\}| \leq 6 -2=4$. Therefore, we also have $\Delta <0$ which is again contradictory as in the proof of Theorem~\ref{l:5/6-majority-3-wise}. 
\end{proof}

By a similar argument, we obtain the following more general consequence.   

\begin{corollary}
\label{c:5/6-majority-rule-consequence}
Let $b,n \in \N$ such that $n \geq 3$ and $(n-b-2)(n-b-6)\leq 3b$. 
Consider an election with $\leq n$ candidates. Let  $x$ be a non-dirty candidate with respect to the $5/6$-majority rule. Let $r \colon A>x>B$ be a median with respect to $d^3_{KT}$ such that $|B| \geq b$ and $z \in A$ whenever $z \geq_{5/6} x$. 
Then for every $y\neq x$ with $x \geq_{5/6} y$, we have $x > y$  in $r$.   \qed 
\end{corollary}

\begin{example}
    In Corollary~\ref{c:5/6-majority-rule-consequence}, we can take the following values for $(n,b)$: 
    \begin{center}
\begin{tabular}{ | M{3em} | M{1cm}  | M{1cm}  | M{1cm} | M{1cm} | M{1cm}| M{1cm}| M{1cm}| M{1cm}|M{1cm}|} 
  \hline
  $n$  & 6 & 8 & 10 & 12 & 14 & 16 & 18 & 20 \\ 
  \hline
  $b$ & 0 & 2 &  3 & 4 &  6 & 7 & 9 & 11\\ 
  \hline
\end{tabular}
\end{center}

\end{example}

The above reduction results are meaningful  since the search space is already very large when $n=14$ as $14!\simeq 8.7\times 10^{10}$. 
The proof of Theorem~\ref{l:5/6-majority-3-wise} is also a simple modification of the proof of Theorem~\ref{l:5/6-majority-3-wise-main}. 
\begin{proof}[Proof of Theorem~\ref{l:5/6-majority-3-wise}]
We adopt again the proof by contradiction of Theorem~\ref{l:5/6-majority-3-wise-main}. 
We deduce from the hypothesis \eqref{e:5/6-majority-3-wise-condition} that $|B|\geq \lambda |L|$. Therefore, as $s >5/6$, we find that 
\[
(5-6s)|L||B|\geq (5-6s)|L| \lambda |L|. 
\]
\par 
Consequently, we obtain from the third line in the inequality \eqref{l:5/6-majority-3-wise-main-delta} that: 
\begin{align}
\label{l:5/6-majority-3-wise-main-delta}
\frac{\Delta}{m} & 
<   (2-3s)  (\vert L \vert +  \vert B \vert)   + (2-2s)\frac{|L|(|L|-1)}{2}  +(5-6s)  |L|.|B| \nonumber\\ 
& \leq (2-3s)  (\vert L \vert +  \vert B \vert)   + (1-s)|L|(|L|-1)  +(5-6s)  |L|.\lambda |L| \nonumber \\
& =  (5\lambda +1 - (6\lambda +1 )s) |L|^2 +(1-2s)|L|  + (2-3s) |B|
\nonumber\\
& \leq 0 \quad \quad\quad  \left(\text{as }s \geq \frac{5\lambda+1}{6\lambda+1}>\frac{5}{6}\right). \nonumber 
\end{align}
which is  a contradiction as in the proof of Theorem~\ref{l:5/6-majority-3-wise}. The proof is thus complete. 
\end{proof}

\bibliography{Kemeny}

\appendix 

\section{Simulations with uniform data}

\begin{table}[h!]
\centering
\caption{Applicability of the Always Theorem (AT), the $2$-wise Extended Always Theorem \ref{t:extended-always} (2AT) and the $3$-wise Extended Always Theorem  \ref{t:3-wise-unanimity-general} (3AT) over 100 000 random   instances uniformly generated with python seed(1) generator}
\label{table:applicability-AT}

\begin{tabularx}{1\textwidth}{ 
  | >{\centering\arraybackslash}X 
  | >{\centering\arraybackslash}X 
  | >{\centering\arraybackslash}X 
  | >{\centering\arraybackslash}X 
  | >{\centering\arraybackslash}X |}
\hline
\textbf{n} & \textbf{m} & \textbf{AT} & \textbf{2AT} & \textbf{3AT}\\ 
\hline  \hline 
3 & 3 & 52.993\% & 52.993\%&  52.993\%  
\\ 
\hline  
- & 4 & 29.916\% & 93.208\%&  29.916\%  
\\ 
\hline  
 - & 7 & 3.452\% & 80.354\%&  30.407\% 
 \\ 
\hline
 - & 10 & 0.419\% & 67.896\%&  27.414\% 
 \\ 
\hline
 - & 13 & 0.058\% & 57.042\%&  23.558\% 
 \\ 
\hline
 - & 16 & 0.005\% & 47.087\%&  5.88\% 
 \\ 
\hline
 - & 19 & 0.0\% & 39.294\%&  5.453\% 
 \\ 
\hline
- & 22 & 0.0\% & 32.635\%&  4.765\% 
 \\ 
\hline
- & 25 & 0.0\% & 26.871\%&  4.09\% 
 \\ 
\hline
- & 28 & 0.0\% & 22.35\%&  1.079\% 
 \\ 
\hline\hline 
4 & 5 & 28.159\% & 89.767\%&  28.159\% 
 \\ 
\hline
- & 9 &  1.713\% & 62.494\%&  18.92\% 
 \\ 
\hline
- & 13 &  0.139\% & 38.778\%&  1.935\% 
 \\ 
\hline
- & 17 &  0.007\% & 22.962\%&  1.244\% 
 \\ 
\hline
- & 21 &  0.001\% & 13.341\%&  0.119\% 
 \\ 
\hline
- & 25 &  0.0\% & 7.652\%&  0.104\% 
 \\ 
\hline\hline 
5 & 6 &  23.712\% & 84.405\%&  23.712\% 
 \\ 
\hline
- & 11 &  0.973\% & 42.36\%&  0.973\% 
 \\ 
\hline
- & 16 &  0.031\% & 17.053\%& 0.531\% 
 \\ 
\hline
- & 21 &  0.0\% & 6.425\%& 0.03\% 
 \\ 
\hline\hline 
10 & 11 &  4.038\% & 33.935\%& 4.038\% 
 \\ 
\hline
- & 15 &  0.282\% & 4.013\%& 0.282\% 
 \\ 
\hline
- & 21 &  0.007\% & 0.95\%& 0.007\% 
 \\ 
\hline\hline 
15 & 16 &  0.327\% & 4.872\%& 0.327\% 
 \\ 
\hline\hline
20 & 21 & 0.017\% & 0.394\% &0.017\%
\\ 
\hline
\end{tabularx}

\end{table}

\end{document}